\documentclass{article}
\usepackage[english]{babel}
\usepackage{amsmath,amssymb,latexsym,theorem}
\usepackage[a4paper, top=2.5cm, bottom=2.5cm, left=1.5cm, right=1.5cm]{geometry}
\usepackage{hyperref}
\usepackage{mathtools, xcolor}
\usepackage{authblk}

\newcounter{nnacknowledgments}

{\theorembodyfont{\rmfamily}\newtheorem{acknowledgments*}[nnacknowledgments]{Acknowledgments}}

\newcommand{\assign}{:=}
\newcommand{\cdummy}{\cdot}
\newcommand{\mathd}{\mathrm{d}}

\newcommand{\tmem}[1]{{\em #1\/}}

\newcommand{\tmrsup}[1]{\textsuperscript{#1}}
\newcommand{\tmtextbf}[1]{\text{{\bfseries{#1}}}}
\newcommand{\tmtextit}[1]{\text{{\itshape{#1}}}}
\newenvironment{proof}{\noindent\textbf{Proof\ }}{\hspace*{\fill}$\Box$\medskip}
\newtheorem{definition}{Definition}
\newtheorem{proposition}{Proposition}
\newtheorem{corollary}{Corollary}
\newcounter{nnproposition}

\newtheorem{proposition*}[nnproposition]{Proposition}
{\theorembodyfont{\rmfamily}\newtheorem{remark}{Remark}}
\newcommand{\Wick}[1]{\mathopen{}\vcentcolon #1\vcentcolon\mathclose{}}
\newcommand{\alg}{\mathrm{alg}}
\newcommand{\phys}{\mathrm{phys}}
\newtheorem{lemma}{Lemma}

\begin{document}

	\title{Two-Dimensional Euclidean Free Fermi Fields from Grassmannian Stochastic Analysis}
	
	\author[1,4,5]{\small Claudio Dappiaggi}
	\author[2,5]{\small Francesco Carlo De Vecchi}
	\author[3,6,7]{\small Pietro Tullio Falzoni}
	\author[1,4,5]{\small Paolo Rinaldi}
	
	\affil[1]{\footnotesize Dipartimento di Fisica Alessandro Volta, Universit{\`a} degli studi di Pavia, Via Bassi 6, I-27100 Pavia, Italia}
	\affil[2]{\footnotesize Dipartimento di Matematica Felice Casorati, Universit{\`a} degli studi di Pavia, Via Ferrata 5, I-27100 Pavia, Italia}
	\affil[3]{\footnotesize Dipartimento di Matematica, Universit{\`a} di Genova, Via Dodecaneso 35, I-16146, Genova, Italia}
	\affil[4]{\footnotesize  INFN -- Sezione di Pavia, Via Bassi 6, I-27100 Pavia, Italia}
	\affil[5]{\footnotesize INdAM -- Sezione di Pavia, Via Ferrata 5, I-27100 Pavia, Italia}
	\affil[6]{\footnotesize INdAM -- Sezione di Genova, Via Dodecaneso 35, I-16146, Genova, Italia}
	\affil[7]{\footnotesize INFN -- Sezione di Genova, Via Dodecaneso 33, I-16146, Genova, Italia}
	
	
	\maketitle
	\abstract
	Motivated by Nelson's Euclidean Markov-field construction, we develop a counterpart for the massive two-dimensional free Dirac field. Starting from a Gaussian white noise with values in a Grassmann algebra on a non-commutative probability space, we construct the associated Brownian motion and a stationary Ornstein--Uhlenbeck field. The resulting Euclidean Fermi field is realized as a tempered distribution on spacetime test functions, with values in the underlying algebra. Its mixed two-point functions reproduce the Euclidean Dirac fundamental solution.\\
	\noindent Time-zero fields are obtained by conditional expectation onto the past filtration generated by the noise, while their covariance is the reflected Euclidean Dirac kernel. We show that conditional expectation is compatible with Wick ordering and maps Euclidean Wick polynomials to the corresponding time-zero Wick polynomials. The Euclidean time-translation semigroup is then constructed on time-zero fields and Wick sectors, with generators given by the one-particle Hamiltonian and by its second quantization.
	{\tableofcontents}
	
	\section{Introduction}\label{Sec: Introduction}

Constructive quantum field theory was originally shaped by the idea that Euclidean fields should not merely be described in terms of functional integrals, but as probabilistic objects. This has the merit of making it possible to establish criteria for a rigorous definition of Wick rotation, the key tool to reconstruct a corresponding Lorentzian theory. In the Bosonic case this viewpoint is embodied by the Osterwalder--Schrader axioms \cite{OS73b} and by Nelson's construction of quantum fields from Markov fields \cite{Nel73}, see also \cite{Sim74, GJ87}. In this last approach one starts from an Euclidean random field rather than from the canonical commutation relations. This allows to identify a distinguished time-zero algebra by means of conditional expectations recovering in addition the Hamiltonian dynamics from the Euclidean time-translation semigroup. \\

\noindent This programme has acquired renewed momentum in recent years thanks to the interaction between constructive quantum field theory and singular stochastic partial differential equations, which are revolving mainly around regularity structures and paracontrolled calculus \cite{Hai14,GIP15}. In turn, establishing these new frameworks has allowed to put on firm mathematical grounds the Parisi--Wu idea of stochastic quantization \cite{PW81}, which allows for the construction and analysis of Euclidean models beyond the perturbative level. Yet, most of the  advances mentioned above concern Bosonic theories, while Fermionic models are less studied. The main structural difference lies in the anti-commuting nature of the underlying variables. This prevents a direct realization of Fermionic Euclidean fields as random variables on an ordinary commutative probability space, as already stressed in the analysis of Osterwalder and Schrader \cite{OS73a}. Consequently, the probabilistic side of the Nelson programme has no immediate Fermionic analogue.\\

\noindent A way around this obstruction has recently emerged in \cite{ABDG22,DFGG25}. Herein the idea consists of replacing the commutative probability space by a non-commutative counterpart. This is a unital $C^*$-algebra endowed with a state, that is a positive and normalized functional thereon. Grassmann random variables are therefore regarded as algebra homomorphisms from an exterior algebra into this non-commutative probability space. In this setting one can define key probabilistic notions ranging from a Grassmannian white noise to conditional expectations. These allow to devise a natural stochastic language to encode anti-commuting variables, that is a framework in which Nelson's construction can be revisited for Fermionic fields. On the one hand, it is worth mentioning that this non-commutative setting has already proved useful in the study of weakly coupled subcritical Euclidean Fermionic field theories \cite{DFG25}. On the other hand, it suggests that this might be the right framework to analyze interacting two-dimensional models, among which the Thirring model is a natural test case, see \cite{BCDR24} for a recent analysis in the framework of singular stochastic partial differential equations.\\

\noindent The purpose of the present paper is to carry out part of this programme, focusing on the non-interacting, massive Euclidean Dirac field in two dimensions. More precisely, we construct a Euclidean free Fermi field starting from a Grassmannian white noise on a non-commutative probability space. This reconstruction may be viewed as a non-commutative fermionic
	analogue of Nelson's construction from Markov fields \cite{Nel73}.
	Here, the Markov structure comes from an auxiliary stationary
	Grassmannian Ornstein--Uhlenbeck field $\varphi$ and its driving white
	noise. An algebra-valued Euclidean Fermi field $\psi$, whose mixed
	two-point functions reproduce the Euclidean Dirac fundamental
	solution, is derived from $\varphi$ by applying a differential
	operator, local in time and space, to its second component. Hence, we
	do not require $\psi$ itself to be Markovian. Subsequently, we exploit the
	Markovianity of the auxiliary field and the associated conditional
	expectation to show that the resulting evolution coincides with the
	second quantization of the one-particle semigroup. A Markov property has been established also in \cite{FO74} in their
	gauge-space construction, where the relevant algebra is generated by
	Euclidean fermionic gauge fields. Their proposal to introduce
	interactions through a normal central change of state encounters an
	obstruction. Our noise-based construction offers instead a different possible
	starting point for interacting models. We restrict the present
	analysis to the free theory so as to isolate the Markovian
	construction from renormalization issues. The main results can be summarized as follows:
\begin{itemize}
\item We construct a Grassmannian white noise and the associated Grassmannian Brownian motion on a non-commutative probability space. The covariance is specified by a bounded operator whose associated two-point function is an antisymmetric bilinear form, as required by the fermionic nature of the variables. The resulting process satisfies the corresponding Gaussian Wick rule and has independent increments.
\item We define a stationary Grassmannian Ornstein--Uhlenbeck field and, from it, a two-dimensional Euclidean free Fermi field. Its only non-vanishing two-point functions are the mixed ones. These coincide with the integral kernel of the Euclidean Dirac propagator.
\item We construct zero-time Fermionic fields by conditional expectation onto the algebra generated by the past noise. Their covariance is the reflected Euclidean Dirac kernel and the conditional expectation is compatible with Wick ordering. In other words, Euclidean Wick polynomials are mapped to the corresponding Wick polynomials of the zero-time fields.
\item We show that Euclidean time translations induce a strongly continuous semigroup on the one-particle zero-time space. Its generator is the positive one-particle Hamiltonian $A_x=(-\partial_x^2+m^2)^{1/2}$. On Wick polynomials the induced generator is the corresponding second quantization.
\end{itemize}

The paper is organized as follows. In the first part of Section~1 we fix the analytic conventions for two-dimensional Euclidean Dirac fields and recall succinctly the Osterwalder--Schrader framework. In Section~2 we introduce the non-commutative probability space, Grassmann random variables, Gaussian covariances and conditional expectations. In Section~3 we discuss the construction of the Grassmannian white noise, the Ornstein--Uhlenbeck field, the Euclidean free Fermi field, the time-zero fields and the Euclidean time-evolution semigroup. The appendix contains the proof of the compatibility between conditional expectation and Wick ordering.

\subsection{Analytic Structures}\label{Sec: Analytic Structures}

In this section we fix the notation and conventions used throughout. We take
as background the two-dimensional Euclidean space $(\mathbb{R}^2, \delta)$,
with signature $(+, +)$ and Cartesian coordinates $(t, x)$. We shall refer to
$t$ as Euclidean time, in contrast to its Lorentzian counterpart $t^M$,
related to it by the Wick rotation
\[ t = - it^M . \]
To simplify the notation we shall denote a generic point of the Euclidean
space by $\mathrm{x}\equiv (t, x)$ and its Lorentzian counterpart by $\mathrm{x}^M \equiv
(t^M, x)$.

\paragraph{Function spaces.} On $\mathbb{R}^2$ we denote by $\mathcal{S}
(\mathbb{R}^2)$ the space of rapidly decreasing smooth functions, endowed with
its standard locally convex Fr{\'e}chet topology. Its topological dual, the
space of tempered distributions, is denoted by $\mathcal{S}' (\mathbb{R}^2)$.
Given $\tau \in \mathbb{R}$, we introduce the subspaces
\begin{equation}
  \mathcal{S}^-_{\tau} (\mathbb{R}^2) \assign \left\{ v \in
  \mathcal{S}(\mathbb{R}^2) \hspace{0.27em} | \hspace{0.27em} v (t, x) = 0,
  \text{for } t > \tau \right\},\quad\textrm{and}\quad\mathcal{S}^+_{\tau} (\mathbb{R}^2) \assign \left\{ v \in
  \mathcal{S}(\mathbb{R}^2) \hspace{0.27em} | \hspace{0.27em} v(t, x) = 0,
  \text{for } t < \tau \right\} . \label{Eq: S-_tau}
\end{equation}
Equivalently, the elements of $\mathcal{S}^+_{\tau} (\mathbb{R}^2)$
(\tmtextit{resp.} $\mathcal{S}^-_{\tau} (\mathbb{R}^2)$) are Schwartz
functions supported in the closed half-space $\{t \geq \tau\}$
(\tmtextit{resp.} $\{t \leq \tau\}$). Since the Fourier transform will play an important role, we fix the convention
\[ \mathcal{F}: \mathcal{S} (\mathbb{R}^2) \longrightarrow \mathcal{S}
   (\mathbb{R}^2), \qquad v (t, x) \longmapsto \hat{v} (\omega, k) =
   \frac{1}{2 \pi}  \int_{\mathbb{R}^2} \mathd t \mathd x\, v(t, x) e^{- i
   (\omega t + kx)} . \]
Fourier transforms with respect to only one of the two variables are defined
analogously.

\paragraph{Euclidean Dirac Operator --} We consider distributional Euclidean
spinors and cospinors, namely elements respectively of
\[ \mathcal{S}' (\mathbb{R}^2 ; \mathbb{C}^2), \quad \mathcal{S}'
   (\mathbb{R}^2 ; (\mathbb{C}^2)^{\ast}) . \]
Throughout the paper, using the standard Hermitian pairing, we shall freely
identify $\mathbb{C}^2$ and $(\mathbb{C}^2)^{\ast}$ whenever necessary and if no
confusion can arise. The Euclidean Dirac operator and its formal adjoint are
\begin{equation}
  D_E = \gamma_E^{\mu} \partial_{\mu} + m : \mathcal{S}' (\mathbb{R}^2 ;
  \mathbb{C}^2) \longrightarrow \mathcal{S}' (\mathbb{R}^2 ; \mathbb{C}^2),
  \qquad D_E^{\ast} = - \gamma_E^{\mu} \partial_{\mu} + m : \mathcal{S}'
  (\mathbb{R}^2 ; \mathbb{C}^2) \longrightarrow \mathcal{S}' (\mathbb{R}^2 ;
  \mathbb{C}^2) . \label{Eq:dirac-op}
\end{equation}
Here $m \geqslant 0$ is a mass parameter, while $\gamma_E^{\mu} \in
\mathcal{M} (2 ; \mathbb{C})$, $\mu = 0, 1$, are Euclidean gamma matrices.
They are Hermitian and satisfy the Clifford relations
\[ \{\gamma_E^{\mu}, \gamma_E^{\nu} \} = 2 \delta^{\mu \nu} \mathbb{I}_2,
   \qquad \mu, \nu = 0, 1. \]
A convenient choice is given by the Pauli matrices
\[ \gamma_E^0 = \sigma^1 = \left(\begin{array}{cc}
     0 & 1\\
     1 & 0
   \end{array}\right), \qquad \gamma_E^1 = \sigma^2 = \left(\begin{array}{cc}
     0 & - i\\
     i & 0
   \end{array}\right) . \]
{\noindent}For later use, we recall the identities

\begin{align*}
  D_E D_E^{\ast} & = (- \Delta + m^2) \mathbb{I}_2,\\
  D_E^{\ast} D_E & = (- \Delta + m^2) \mathbb{I}_2,
\end{align*}
where $\Delta \assign \partial_t^2 + \partial_x^2$ is the Euclidean Laplacian
on $\mathbb{R}^2$, while $\mathbb{I}_2$ denotes the $2 \times 2$ identity
matrix. For $m > 0$, we denote by $G \in \mathcal{S}' (\mathbb{R}^2)$ the
translation-invariant fundamental solution of $- \Delta + m^2$ which decays at
infinity. Its integral kernel reads
\[ G (t, x) = \frac{1}{2 \pi} K_0 \left( m \sqrt{t^2 + x^2} \right), \]
where $K_0$ is the modified Bessel function of second kind of order $0$.
Equivalently, in Fourier space, we have
\[ \hat{G} (\omega, k) = \frac{1}{2 \pi}  \frac{1}{\omega^2 + k^2 + m^2} . \]
Observe that, in the massless case $m = 0$, no fundamental solution vanishing
at infinity exists in two dimensions and one must instead choose an infrared
prescription. We avoid entering into the details leaving this choice
understood. Consequently, the Euclidean Dirac fundamental solution is the
matrix-valued tempered distribution
\begin{equation}\label{Eq: Sol-fond}
  S_E \assign D_E^{\ast}  (G\mathbb{I}_2) \in \mathcal{S}' (\mathbb{R}^2 ;
  \mathrm{End} (\mathbb{C}^2)), 
\end{equation}
whose expression in Fourier space is

\begin{align*}
  & \hat{S}_E (\omega, k) = \frac{1}{2 \pi}  \frac{- i \gamma_E^0 \omega - i
  \gamma_E^1 k + m\mathbb{I}_2}{\omega^2 + k^2 + m^2} .
\end{align*}

Observe that this expression is valid also for $m = 0$ since, in this case,
$S_E$ is locally integrable in a neighbourhood of $(\omega, k) = (0,
0)$.

\subsection{Osterwalder--Schrader Distributions and Reconstruction}

In this subsection we recall the Osterwalder--Schrader framework for a two-dimensional Bosonic scalar theory. We follow mainly {\cite{OS73b}} and
{\cite{Sim74}}. We omit the proofs, in particular of the reconstruction
theorem, since they can be found in the mentioned references. Actually our
main goal is to fix the functional-analytic setting and to state the class of
Euclidean correlation functions to which the theorem applies, as a preparation
for the discussion of the Fermionic counterpart in Section \ref{Sec:
Construction of Fermionic Fields}. We identify the \tmtextit{large diagonal}
$$\Delta_n \assign \left\{ (\mathrm{x}_1, \cdots, \mathrm{x}_n) \in \mathbb{R}^{2 n}\hspace{0.27em} | \hspace{0.27em} \mathrm{x}_i = \mathrm{x}_j\;  \text{for some } i \neq j \right\} .$$
At the same time we set
$$\mathcal{S}_n \assign \mathcal{S} (\mathbb{R}^{2 n}) \quad \mathrm{and}\quad^0 \mathcal{S}_n \assign \left\{ u \in \mathcal{S}_n \hspace{0.27em} | \hspace{0.27em} \partial^{\alpha} u |_{\Delta_n} = 0, \text{for every multiindex } \alpha \right\},$$
where $^0 \mathcal{S}_n$ is a closed subspace of $\mathcal{S}_n$, see
{\cite{OS73b}}.

\begin{remark}
	These spaces come accompanied with their topological duals $\mathcal{S}^\prime_n
	\equiv \mathcal{S}' (\mathbb{R}^{2 n})$ and $(^0 \mathcal{S}_n)^\prime$. Although
	in {\cite{OS73b}} the latter is considered, in this work we shall assume
	from the outset that the Schwinger functions are represented by tempered
	distributions $S_n\in\mathcal{S}^\prime_n$ satisfying the Osterwalder--Schrader
	axioms.  If one starts instead from an element lying in $(^0 \mathcal{S}_n)'$, an
	extension to $\mathcal{S}^\prime_n$ exists by the Hahn-Banach theorem. Yet, such extension is nonunique and it might not preserve the structural properties encoded by the Osterwalder-Schrader axioms unless these are imposed. Nonetheless, the ambiguity is supported on the large diagonal $\Delta_n$ and, once a compatible extension has been fixed, this does not affect the reconstruction.
\end{remark}

{\noindent}Two additional ingredients are needed:
\begin{enumerate}
  \item the \tmtextit{Euclidean time reflection}
  \begin{equation}\label{Eq: Time-Reflection}
  \vartheta : \mathbb{R}^2 \to \mathbb{R}^2, \quad \mathrm{x}=(t, x) \mapsto \vartheta\mathrm{x}= (- t, x),
  \end{equation}
  which extends to $\mathcal{S}_n$ as the map
$$\theta_n : \mathcal{S}_n \to \mathcal{S}_n, \qquad u(\mathrm{x}_1, \cdots, \mathrm{x}_n) \mapsto (\theta_n u) (\mathrm{x}_1, \cdots, \mathrm{x}_n) = u(\vartheta (\mathrm{x}_n),\cdots, \vartheta (\mathrm{x}_1)).$$
  \item the \tmtextit{positive-time test functions}. Given $\mathbb{R}^2_+
  \assign [0, \infty) \times \mathbb{R} \ni (t, x)$, we call
 \begin{equation}\label{Eq: Positive time Test Function Space}
  	\mathcal S(\mathbb R^2_+)
  	\assign \mathcal S^+_0(\mathbb R^2)
  	=\{f\in\mathcal S(\mathbb R^2):
  	\operatorname{supp}f\subset[0,\infty)\times\mathbb R\},
  \end{equation}  
   where $\mathcal{S}^+_0 (\mathbb{R}^2)$ is defined as per Equation \eqref{Eq: S-_tau}. For any $n \geqslant 1$, we set
  \[ \mathcal{S}_{+, n} \assign \widehat{\bigotimes}_{j = 1}^n \mathcal{S}
     (\mathbb{R}^2_+), \]
  together with
  \[ \underline{\mathcal{S}} \assign \mathbb{C} \oplus \bigoplus_{n \geqslant
     1}^{\mathrm{alg}} \mathcal{S}_n, \qquad \underline{\mathcal{S}}_+ \assign
     \mathbb{C} \oplus \bigoplus_{n \geqslant 1}^{\mathrm{alg}} \mathcal{S}_{+, n} .
  \]
 The superscript $\mathrm{alg}$ denotes the algebraic direct sum, namely only sequences with finitely many nonzero components are admitted.
\end{enumerate}
\begin{definition}[Bosonic Osterwalder--Schrader distributions]\label{Def: Bosonic OS Distributions} 
	A collection
  \[ (S_n)_{n \in \mathbb{N}_0}, \qquad S_0 = 1, \qquad S_n \in
     \mathcal{S}'_n, \]
  is called a family of Bosonic Osterwalder-Schrader distributions if
  \begin{enumerate}
    \item[(OS1)] {\tmem{Euclidean covariance --}} For every element of the
    proper Euclidean isometry group $(a,R) \in E (2)^+ =\mathbb{R}^2 \ltimes SO(2)$, it holds
    \[ S_n (u) = S_n (u_{(a, R)}), \qquad u\in \mathcal{S}_n, \]
    where
    \[ u_{(a,R)} (\mathrm{x}_1, \cdots, \mathrm{x}_n) \assign u(R^{- 1} (\mathrm{x}_1 - a), \ldots,
       R^{- 1} (\mathrm{x}_n - a)) . \]
    \item[(OS2)] {\tmem{Reflection positivity --}} For every finite sequence
    $$(u_0, u_1, \ldots, u_N) \in \underline{\mathcal{S}}_+,$$
    one has
    $$\sum_{n, m = 0}^N S_{n + m}  (\overline{\theta_n u_n} \otimes u_m) \geq 0.$$
    The terms with $n = 0$ or $m = 0$ are understood in the natural way,
    employing $S_0 = 1$.
    
    \item[(OS3)] {\tmem{Symmetry.}} Denoting by $\mathfrak{S}_n$, $n \geq 1$,
    the permutation group of $n$ elements, for every $\pi \in \mathfrak{S}_n$,
    $$ S_n (u^{\pi}) = S_n (u), \qquad u \in \mathcal{S}_n, $$
    where
  $$u^{\pi} (\mathrm{x}_1, \cdots, \mathrm{x}_n) \assign u(\mathrm{x}_{\pi (1)}, \cdots, \mathrm{x}_{\pi (n)}).$$
    \item[(OS4)] {\tmem{Cluster property --}} Given $a \equiv (a, 0) \in E
    (2)^+$ a non-vanishing spatial translation, if $T_{\lambda a}$ denotes the
    simultaneous translation of all arguments of a test function by $\lambda
    a$, then
$$ \lim_{\lambda \to + \infty} S_{n + m}  (u \otimes T_{\lambda a} u^\prime) = S_n (u) S_m (u^\prime),$$
    for all test functions $u \in \mathcal{S}_n$ and $u^\prime \in \mathcal{S}_m$.
  \end{enumerate}
\end{definition}

\begin{remark}
  The relevance of the preceding definition is fully justified by the
  celebrated Osterwalder-Schrader reconstruction theorem {\cite{OS73b}}. We
  avoid stating it since it would require a long detour on the formulation of
  the Wightman axioms for a quantum field theory on Minkowski spacetime
  {\cite{SW64}}. We limit ourselves to observing that, given a family
  $(S_n)_{n \in \mathbb{N}_0}$ of Bosonic Osterwalder--Schrader distributions
  as per Definition \ref{Def: Bosonic OS Distributions}, then there exists a
  unique sequence $(W_n)_{n \in \mathbb{N}_0}$, $W_n \in \mathcal{S}^\prime(\mathbb{R}^{2 n})$, known as Wightman distributions such that, in addition
  to a set of structural properties, they are obtained from $(S_n)_{n\in\mathbb{N}_0}$ as boundary values of their analytic continuation in the Euclidean time variables. Hence, up to unitary equivalence, they identify a Hilbert space
  $\mathcal{H}$, a normalized vector $\Omega \in \mathcal{H}$, a unitary representation of the proper orthochronous Poincar{\'e} group
  satisfying the spectrum condition and an operator valued distribution $\phi$
  such that
$$W_n (\mathrm{x}^M_1, \cdots, \mathrm{x}^M_n) = \langle \Omega, \phi (\mathrm{x}^M_1) \cdots \phi(\mathrm{x}^M_n) \Omega \rangle.$$
\end{remark}

We conclude this short section remarking that the Osterwalder-Schrader
approach is well understood for Bosonic fields also in view of the work by
Nelson {\cite{Nel73}} who established a consistent method for defining Bosonic
Euclidean fields as random variables on a suitable probability space. However,
the scenario for Fermionic fields is more delicate. A notable attempt to
address this issue can be ascribed to Osterwalder and Schrader {\cite{OS73a}},
who investigated whether Fermionic correlation functions could be derived from
a genuine underlying probability measure. Their analysis highlighted a
fundamental obstruction, namely the anti-commutativity of Grassmann variables
prevents the definition of positive-definite densities. Consequently, unlike
the Bosonic case, Fermionic Euclidean fields lack a realization in terms of
Nelson's stochastic framework and we shall address this issue in the following
sections.

\section{Non-commutative probability space and Grassmann random variables}\label{Sec: Non-commutative probability space}

In this section we introduce Grassmann random variables and we start by
assigning the underlying building blocks:
\begin{enumerate}
  \item a \tmtextit{non-commutative probability space} as introduced in
  {\cite{ABDG22}}, that is a pair $(\mathcal{A}, \omega)$, where $\mathcal{A}$
  is a unital $C^{\ast}$-algebra while $\omega : \mathcal{A} \to \mathbb{C}$
  is a linear, normalized, positive functional, known as a state,
  
  \item a complex, separable Hilbert space $(V, (\cdot, \cdot)_V)$, which we
  assume anti-linear in the left variable.
\end{enumerate}
\begin{definition}
  \label{Def: Grassmann Random Variable}Denoting by $\Lambda V$ the exterior
  algebra associated to $V$, we denote by $\mathcal{G} (V)$ the collection of
  all \tmtextbf{Grassmann random variables}, namely unital algebra
  homomorphisms $X : \Lambda V \to \mathcal{A}$. Given $X \in \mathcal{G} (V)$
  its {\tmem{law}} is the family of moments
  \[ \omega^X (F) \assign \omega (X (F)), \qquad \forall F \in \Lambda V. \]
  Furthermore we say that $X \in \mathcal{G} (V)$ is continuous if
  \[ \| X \|_{\mathcal{G} (V)} \assign \sup_{\| v \|_V = 1} \| X (v)
     \|_{\mathcal{A}} < \infty . \]
  The collection of all such random variables is denoted by $\mathcal{G}^{-
  \infty} (V)$.
\end{definition}

Observe that one cannot expect $\mathcal{G} (V)$ to be
complete, unless $V$ is finite dimensional {\cite[Lemma 7]{ABDG22}}. The next
step consists of establishing a notion of compatibility.

\begin{definition}[Compatible Grassmann random variables]
  \label{Def: Compatible Grassmann Random Variables} Given two complex,
  separable Hilbert spaces $V$ and $W$ as well as $(\mathcal{A}, \omega)$, a
  non-commutative probability space, we say that two Grassmann random
  variables $X \in \mathcal{G} (V)$ and $Y \in \mathcal{G} (W)$ are
  \tmtextbf{compatible} if the linear map $Z : V \oplus W \rightarrow
  \mathcal{A}$ such that $Z ((v, 0)) = X (v)$, if $v \in V$, while $Z ((0, w))
  = Y (w)$, if $w \in W$, extends to a Grassmann random variable $Z : \Lambda
  (V \oplus W) \rightarrow \mathcal{A}$.
\end{definition}

\begin{remark}\label{remark:independence}
As a consequence of this last definition we can establish a notion of
{\tmem{independence}}. More precisely, two compatible Grassmann random
variables $X \in \mathcal{G} (V, \mathcal{A})$ and $Y \in \mathcal{G} (W,
\mathcal{A})$ as per Definition \ref{Def: Compatible Grassmann Random
Variables} are said to be independent with respect to the state $\omega$ if,
for all $F_1 \in \Lambda (V)$ and $F_2 \in \Lambda (W)$
\begin{equation}
  \label{Eq: Independence of GRV} \omega (X (F_1) Y (F_2)) = \omega (X (F_1))
  \omega (Y (F_2)) .
\end{equation}
\end{remark}

\subsection{Gaussian Grassmann random variables}\label{Sec: Gaussian Grassmann random variables}

We enrich the framework considered up to now, assuming that $V$ comes endowed
also with a \tmtextit{real structure} $\chi : V \rightarrow V$, compatible
with $(\cdummy, \cdummy)_V$, that is an anti-linear isometry such that
\[ (\chi (v), \chi (w))_V = (w,v)_V, \quad \chi^2
   (v) = v, \quad \forall v, w \in V. \]
On the one hand, this induces a bilinear form
\[ \langle \hspace{-0.17em} \langle \cdot, \cdot \rangle \hspace{-0.17em}
   \rangle_{\chi} : V \times V \rightarrow \mathbb{C}, \quad (v, w) \mapsto
   \langle \hspace{-0.17em} \langle v, w \rangle \hspace{-0.17em}
   \rangle_{\chi} = (\chi (v), w)_V . \]
On the other hand, given a bounded linear operator $C : V \rightarrow V$ and
denoting by $C^{\ast}$ its adjoint with respect to $(\cdummy, \cdummy)_V$, we can
define its $\chi$-transpose as
\[ C^{\chi} \assign \chi C^{\ast} \chi, \]
such that
\[ \langle \hspace{-0.17em} \langle v, Cw \rangle \hspace{-0.17em}
   \rangle_{\chi} = \langle \hspace{-0.17em} \langle C^{\chi} v, w \rangle
   \hspace{-0.17em} \rangle_{\chi}, \quad \forall v, w \in V. \]
We say that $C$ is $\chi$-antisymmetric if $C^{\chi} = - C$. On top of these
data, we consider the antisymmetric Fock space $\mathcal{F}_- (V)$, that is
the Hilbert space completion of the exterior algebra $\Lambda V$ with respect
to the scalar product naturally induced from $(\cdummy, \cdummy)_V$, see
{\cite{BR97}}. Thereon, we define
\begin{enumerate}
  \item the {\tmem{creation operator}} $a^{\dagger} : V \rightarrow
  \mathcal{B} (\mathcal{F}_- (V))$ whose action is completely determined by
  \[ a^{\dagger} (v)  (u_1 \wedge u_2 \wedge \ldots \wedge u_k) = v \wedge u_1
     \wedge u_2 \wedge \ldots \wedge u_k, \quad \forall v, u_1, \ldots, u_k
     \in V, \]
  \item the {\tmem{annihilation operator}} $a : V \rightarrow \mathcal{B}
  (\mathcal{F}_- (V))$ whose action is completely determined by
  \[ a (v)  (u_1 \wedge u_2 \wedge \ldots \wedge u_k) = \sum_{j = 1}^k (-
     1)^{j + 1} (v, u_j)_V  \hspace{0.17em} \hspace{0.17em} u_1 \wedge \ldots
     \wedge \hat{u}_j \wedge \ldots \wedge u_k, \quad \forall v, u_1, \ldots,
     u_k \in V, \]
  where $\hat{u}_j$ indicates that the $j$-th element is omitted.
\end{enumerate}
Observe that, being $\mathcal{B} (\mathcal{F}_- (V))$ a
C\tmrsup{$\ast$}-algebra, both $a$ and $a^{\dagger}$ can be extended as
algebra homomorphisms so to identify a Grassmann random variable in
$\mathcal{G} (V, \mathcal{B}(\mathcal{F}_- (V)))$ as per Definition \ref{Def:
Grassmann Random Variable}.

\begin{definition}\label{Def: Creator Grassmann Random Variable} 
Given a complex, separable
  Hilbert space $(V, (\cdot, \cdot)_V)$ endowed with a real structure $\chi :
  V \rightarrow V$ as well as a bounded, $\chi$-antisymmetric linear operator $C
  : V \to V$, we call \tmtextbf{Gaussian, Grassmann random variable} the map
  $X : \Lambda V \to \mathcal{B} (\mathcal{F}_- (V))$ constructed as the
  algebra homomorphism induced by
  \[ X (v) = a^\dagger(Cv) + a (\chi v), \quad \forall v \in V. \]
\end{definition}

A direct application of the defining properties of the creation and
annihilation operators shows that, for any $n \in \mathbb{N}$ and $v_1,
\ldots, v_{2 n} \in V$, the odd moments of $X$ vanish, while the even ones are

\begin{align}
  & \omega (X (v_1) X (v_2) \ldots X (v_{2 n})) = \sum_{\sigma \in
  \mathcal{P}_{2 n}} \mathrm{sgn} (\sigma)  \prod_{i = 1}^n  \langle \hspace{-0.17em} \langle v_{\sigma
  (2 i - 1)}, Cv_{\sigma (2 i)} \rangle \hspace{-0.17em}
  \rangle_{\chi},  \label{eq:Wickrule}
\end{align}

where $\mathcal{P}_{2 n}$ denotes the set of ordered pairings of $\{1, \ldots,
2 n\}$, namely $\sigma (2 i - 1) < \sigma (2 i)$ and $\sigma (1) < \sigma (3)
< \ldots < \sigma (2 n - 1)$.

More generally, Equation \eqref{eq:Wickrule} is taken as the defining identity
for a {\tmem{Gaussian Grassmann random variable}} and one can read the
construction above as a constructive realization of {\cite[Lemma 21]{ABDG22}}.
In other words, for every bounded, $\chi$-antisymmetric operator $C$, there
exists a non-commutative probability space $(\mathcal{A}, \omega)$ and a
Gaussian Grassmann random variable whose covariance is uniquely determined by
$C$.

\begin{remark}
  The existence of a continuous Gaussian Grassmann variable for any bounded
  covariance $C$ is a distinctive feature of the anti-commutative setting. In
  the commutative (Bosonic) setting, there is no counterpart to {\cite[Lemma
  21]{ABDG22}} within the framework of bounded operators. Specifically, for a
  positive, self-adjoint operator $B$ on a Hilbert space $V$, a linear map
  thereon, satisfying the Gaussian moment condition does not necessarily take values in a
  $C^{\ast}$-algebra $\mathcal{A}$.
\end{remark}

\paragraph{Grassmannian conditional expectation -- } To conclude the section
we discuss the notion of {\tmem{conditional expectation}} in the
non-commutative setting. Contrary to the commutative framework, this cannot be
seen as a byproduct of the Radon-Nikodym theorem {\cite[Theorem 6.10]{Rud87}},
which exists \tmtextit{a priori}. A case by case analysis is thus necessary.
In the following we formalize this statement.

\begin{definition}[Conditional expectation]\label{Def: Conditional Expectations}
  Given a non-commutative probability space $(\mathcal{A}, \omega)$ and a
  unital C\tmrsup{$\ast$}-subalgebra $\mathcal{A}' \subseteq \mathcal{A}$, a
  linear map $\varphi : \mathcal{A} \rightarrow \mathcal{A}'$ is a
  \tmtextbf{conditional expectation} of $\mathcal{A}$ onto $\mathcal{A}'$ with
  respect to $\omega$ if it is positive and if, denoting by
  $\mathbb{I}_{\mathcal{A}}$ and $\mathbb{I}_{\mathcal{A}'}$ the unit elements
  respectively of $\mathcal{A}$ and $\mathcal{A}'$,
  \[ \varphi (\mathbb{I}_{\mathcal{A}}) =\mathbb{I}_{\mathcal{A}'}, \quad
     \varphi (B_1 AB_2) = B_1 \varphi (A) B_2, \quad \omega (\varphi (A)) =
     \omega (A), \quad \forall A \in \mathcal{A} \hspace{1.27em} \mathrm{and}
     \quad \hspace{0.27em} \forall B_1, B_2 \in \mathcal{A}' . \]
\end{definition}

As shown in {\cite[Proposition 2.12]{DFG25}}, if such a conditional
expectation exists, it is a $\ast$-preserving projection of unit norm onto
$\mathcal{A}'$. In the following, we are interested in physical systems
governed by a suitable dynamics, which is implemented in this framework by a
family of conditional expectations. The following definition sets the stage to
make this idea precise.

\begin{definition}[Filtration]
  Let $(\mathcal{A}, \omega)$ be a non-commutative probability space and let
  $I \subseteq \mathbb{R}$ be an interval. A {\tmem{filtration}} over $I$ is a
  family $(\mathcal{A}_t)_{t \in I}$ of $C^{\ast}$-subalgebras of
  $\mathcal{A}$ such that $\mathcal{A}_t \subseteq \mathcal{A}_{t'} \subseteq
  \mathcal{A}$ whenever $t \leq t'$. The triple $(\mathcal{A}, \omega,
  (\mathcal{A}_t)_{t \in I})$ is called a \tmtextbf{filtered} non-commutative
  probability space.
\end{definition}

\begin{remark}
  In our analysis we combine all the above structures considering
  {\tmem{adapted Grassmannian processes}} $(X_t)_{t \in I}$. These are
  families of Grassmann random variables as per Definition \ref{Def: Grassmann
  Random Variable} such that the values of $X_t$ lie in the algebra $\mathcal{A}_t$ for every $t \in I$. We
  can construct them as follows: starting from a Grassmann random variable
  $X$, we define $(X_t)_{t \in I}$ in terms of a family of conditional
  expectations $(\omega_t)_{t \in I}$ such that, for every $s \leqslant t$ and
  for every $X \in \mathcal{A}$,
  \[ \omega_s (\omega_t (X)) = \omega_s (X) . \]
  The existence of such a family $(\omega_t)_{t \in I}$ is not guaranteed in
  general, but it is in the setting of Definition \ref{Def: Creator Grassmann
  Random Variable}, where $\mathcal{A} \equiv \mathcal{B} (\mathcal{F}_-
  (V))$. This is discussed thoroughly in {\cite[Appendix A]{DFG25}}.
\end{remark}

\section{Construction of Fermionic fields}\label{Sec: Construction of
Fermionic Fields}

The scope of this section is to present an explicit Fermionic field
construction which is analogous to {\cite{Nel73}}, but realized within the
setting of Grassmannian probability. While the traditional Osterwalder-Schrader program
{\cite{OS73a, OS73b}} illustrated the limitations of classical measure theory, these
obstacles can be bypassed by replacing the standard Kolmogorov axioms with a
non-commutative probabilistic structure. For the sake of clarity, we shall
restrict our construction to the two-dimensional case, comprising one spatial
and one temporal dimension. While this choice is technically
immaterial for the free field, where the Grassmannian Gaussian structure
remains formally invariant, it provides a more transparent framework for
discussing potential extensions to interacting theories. Here the
specificity of the $(1 + 1)$-dimensional spacetime often plays a crucial role
in the regularity and renormalizability of the model.\\
The strategy is the following: Starting from a Grassmannian white noise, we construct the Euclidean and time-zero fields,
identifying the physical one-particle and Fock spaces. Using the Markovian structure of this explicit realization, we exploit the conditional expectation $\omega_0$, which maps translated Euclidean Wick polynomials to time-zero counterparts. We show that the resulting evolution coincides with the second quantization of the one-particle semigroup, thereby identifying the Fock Hamiltonian.

\subsection{Free fields}

We start by fixing the spaces on which the construction is performed. Taking into account the notation and conventions of Section \ref{Sec: Non-commutative probability space}, we set 
$$V:=L^2(\mathbb R^2;\mathbb C^4),$$
and, for future reference, we also consider $E:=\mathcal S(\mathbb R^2;\mathbb C^4)$ seen as a nuclear Fr\'echet space, continuously and densely embedded in $V$. Consequently we can choose as underlying non-commutative probability space the pair $(\mathcal{A},\omega)$, where $\mathcal A:=\mathcal B(\mathcal{F}_-(V))$, while $\omega$ is the vacuum state thereon. Up to an immaterial choice of a basis for $\mathbb{C}^4$, we can split $V$ as $V_1\oplus V_2$ where 
$$V_1\simeq V_2\simeq L^2(\mathbb R^2;\mathbb C^2).$$
In addition, following Section \ref{Sec: Gaussian Grassmann random variables}, we consider $\chi$, the real structure on $V$ induced by the standard complex conjugation, as well as the $\chi$-antisymmetric bounded linear operator $C:V\to V$ such that
\begin{equation}\label{Eq: Covariance C}
C:=\mathbb I_{L^2(\mathbb R^2)}\otimes J\quad\textrm{where}\quad J:=
\begin{pmatrix}
	0 & \mathbb I_{\mathbb C^2}\\
	-\mathbb I_{\mathbb C^2} & 0
\end{pmatrix}.
\end{equation}
These data identify a bilinear map 
\begin{equation}\label{Eq: Bilinear map}
b:V\times V\to\mathbb{C},\quad (v,w)\mapsto b(v,w):=\langle \hspace{-0.17em} \langle v, Cw \rangle \hspace{-0.17em}
\rangle_{\chi},
\end{equation}
such that, if $v=(v^1,v^2)\in V_1\oplus V_2$ and $w=(w^1,w^2)\in V_1\oplus V_2$, then
$$b(v,w)=(\chi(v^1),w^2)_{L^2(\mathbb R^2;\mathbb C^2)}-(\chi(v^2),w^1)_{L^2(\mathbb R^2;\mathbb C^2)}.$$

\begin{definition}[Grassmannian white noise]\label{Def: Grassmannian White Noise}
Given $V=L^2(\mathbb R^2;\mathbb C^4)$ and $\mathcal A:=\mathcal B(\mathcal{F}_-(V))$, a Grassmann random variable $\xi\in \mathcal G^{-\infty}(V,\mathcal A)$ as per Definition \ref{Def: Grassmann Random Variable} is called a {\bf Grassmannian white noise} with covariance $C$ as in Equation \eqref{Eq: Covariance C}, if it is Gaussian, as per Definition \ref{Def: Creator Grassmann Random Variable}, and if it satisfies
$$\omega(\xi(v)\xi(w))=b(v,w),\qquad v,w\in V.$$
\end{definition}

\begin{remark}\label{Rem: Algebra Valued White Noise}
Exploiting that $E\hookrightarrow V$, we can also realize a Grassmannian white noise $\xi$ as an $\mathcal{A}$-valued tempered distribution, that is $\xi\in\mathcal{L}(\mathcal{S}(\mathbb{R}^2;\mathbb{C}^4);\mathcal{A})\equiv\mathcal{S}^\prime(\mathbb{R}^2;\mathbb{C}^4\otimes\mathcal{A})$. Furthermore, on account of the decomposition $V=V_1\oplus V_2$ and given any $F\in E$, we can write
\begin{equation}\label{Eq: Splitting}
\xi=(\xi^1,\xi^2)\quad\textrm{and}\quad F=(f_1,f_2),
\end{equation}
where $F=(f_1,f_2)\in E$. Observe that, henceforth, we shall employ the following notation: Capital letters refer to $\mathbb{C}^4$-valued test-functions, while lower-case letters denote $\mathbb{C}^2$-valued ones. With the covariance as per Equation \eqref{Eq: Covariance C}, the only non-vanishing two-point correlation functions are the mixed ones, namely for $F=(f_1,f_2), F^\prime=(f^\prime_1,f^\prime_2)\in E$, 
$$\omega\bigl(\xi^1(f_1)\xi^2(f'_2)\bigr)=(\chi f_1,f'_2)_{L^2(\mathbb R^2;\mathbb C^2)},\quad\omega\bigl(\xi^2(f_2)\xi^1(f'_1)\bigr)=-(\chi f_2,f'_1)_{L^2(\mathbb R^2;\mathbb C^2)},$$
whereas
$$\omega\bigl(\xi^1(f_1)\xi^1(f'_1)\bigr)=0,\qquad\omega\bigl(\xi^2(f_2)\xi^2(f'_2)\bigr)=0.$$
\end{remark}

We have all the ingredients to establish an analogue of the Brownian motion associated with the Grassmannian white noise. As a preliminary step, we fix the relevant notation. In the following, since the Euclidean time coordinate plays a distinguished role, we wish to be able to clearly disentangle its role from that of the other coordinate $x$. For this reason we set $V_x:=L^2(\mathbb R;\mathbb C^4)$ and we define $C_x:=\mathbb I_{L^2(\mathbb R)}\otimes J.$ Starting from Equation \eqref{Eq: Bilinear map}, we can induce from $b$ a corresponding spatial bilinear map:
\begin{equation}\label{Eq: Spatial Bilinear Covariance}
b_x(F,F^\prime):=(\chi F,C_xF^\prime)_{V_x},\qquad F,F^\prime\in V_x.
\end{equation}

\begin{proposition}\label{Prop: Grassmannian Brownian motion}
	Given a Grassmannian white noise as per Definition \ref{Def: Grassmannian White Noise} on $V\equiv L^2(\mathbb R;V_x)$ with covariance $C$ as in Equation \eqref{Eq: Covariance C}, for $t\geq 0$ and for $F\in V_x$, set
	\begin{equation}\label{Eq: Grassmannian Brownian Motion}
	B_t(F):=\xi(\mathbf 1_{[0,t]}\otimes F),
	\end{equation}
	where $\mathbf{1}_{[0,t]}$ is the characteristic function on the interval $[0,t]$. Hence $(B_t)_{t\geq 0}$ is a Brownian motion, namely, for all $F,F^\prime\in V_x$,
	\begin{enumerate}
		\item it is a centered Gaussian Grassmannian process with
		covariance
		$$\omega(B_t(F)B_s(F^\prime))=\min\{t,s\}\,b_x(F,F^\prime).$$
		\item it has independent increments: for \(0\leq s\leq t\)
		$$B_t(F)-B_s(F)	= \xi(\mathbf 1_{(s,t]}\otimes F),$$
		and
		$$\omega\bigl((B_t(F)-B_s(F))(B_t(F^\prime)-B_s(F^\prime))\bigr)=	(t-s)b_x(F,F^\prime).$$
		\item 	For any two disjoint bounded intervals \(I,J\subset[0,\infty)\)
		$$\omega\bigl(\xi(\mathbf 1_I\otimes F)\,\xi(\mathbf 1_J\otimes F^\prime)\bigr)=0.$$
	\end{enumerate}
\end{proposition}

\begin{proof}
We observe that Equation \eqref{Eq: Grassmannian Brownian Motion} is well-posed since $\mathbf 1_{[0,t]}\otimes F\in V$ for all finite values of $t$. Furthermore, per construction, $B_0(F)=0$ and $B_t$ inherits from $\xi$ the property of being Gaussian and centered. Bearing in mind these observations,
\begin{enumerate}
	\item For any $s,t\geq 0$, Definition \ref{Def: Grassmannian White Noise} entails that, for all $F,F^\prime \in V_x$,
	$$\begin{aligned}
		\omega(B_t(F)B_s(F^\prime))&=\omega\left(\xi(\mathbf 1_{[0,t]}\otimes F)\xi(\mathbf 1_{[0,s]}\otimes F^\prime)\right)=b(\mathbf 1_{[0,t]}\otimes F,\mathbf 1_{[0,s]}\otimes F^\prime)\\
		&=\int_{\mathbb R}\mathbf 1_{[0,t]}(\lambda)\mathbf 1_{[0,s]}(\lambda)\,d\lambda\;b_x(F,F^\prime)=\min\{t,s\}\,b_x(F,F^\prime).
	\end{aligned}$$
	\item If $0\leq s\leq t$, being $\xi$ linear, Equation \eqref{Eq: Grassmannian Brownian Motion} entails
	$$B_t(F)-B_s(F)=\xi\bigl((\mathbf 1_{[0,t]}-\mathbf 1_{[0,s]})\otimes F\bigr)=\xi(\mathbf 1_{(s,t]}\otimes F).$$
	Therefore, for all $F,F^\prime \in V_x$, once more Definition \ref{Def: Grassmannian White Noise} yields
	\begin{gather*}
		\omega\bigl((B_t(F)-B_s(F))(B_t(F^\prime)-B_s(F^\prime))\bigr)=\omega\bigl(\xi\bigl(\mathbf 1_{(s,t]}\otimes F\bigr)\xi\bigl(\mathbf 1_{(s,t]}\otimes F^\prime\bigr)\bigr)=\\
		b(\mathbf 1_{(s,t]}\otimes F,\mathbf 1_{(s,t]}\otimes F^\prime)=\int_{\mathbb R}\mathbf 1_{(s,t]}(\lambda)\,d\lambda\; b_x(F,F^\prime)=(t-s)b_x(F,F^\prime).
	\end{gather*}
	
	\item if $I,J\subset[0,\infty)$ are two disjoint intervals, then $\mathbf{1}_I(\lambda)\mathbf{1}_J(\lambda)=0$. Hence for all $F,F^\prime \in V_x$, Definition \ref{Def: Grassmannian White Noise} entails
$$\omega\bigl(\xi(\mathbf 1_I\otimes F)\,\xi(\mathbf 1_J\otimes F^\prime)\bigr)=b(\mathbf{1}_I\otimes F,\mathbf{1}_J\otimes F^\prime)=\int_{\mathbb R}\mathbf{1}_I(\lambda)\mathbf{1}_J(\lambda)\,d\lambda\; b_x(F,F^\prime)=0.$$
\end{enumerate}
This entails the sought result.
\end{proof}

\begin{remark}
A consequence of point 2 of Proposition \ref{Prop: Grassmannian Brownian motion} is that the Brownian motion $B_t$ has independent increments in the sense of Remark \ref{remark:independence}. More precisely, the compatible Grassmann random variables induced by the maps
\[
V_x\ni F\longmapsto \xi(\mathbf 1_I\otimes F),
\qquad
V_x\ni F\longmapsto \xi(\mathbf 1_J\otimes F),
\]
satisfy Equation~\eqref{Eq: Independence of GRV}. \\
By linearity, it is sufficient to prove Equation~\eqref{Eq: Independence of GRV} for monomials. Let $n,m\in\mathbb N_0$ and let $F_1,\ldots,F_n,G_1,\ldots,G_m\in V_x$. Since $\xi$ is Gaussian, the Wick
rule expresses
\[
\omega\left(
\prod_{r=1}^{n}\xi(\mathbf 1_I\otimes F_r)
\prod_{\ell=1}^{m}\xi(\mathbf 1_J\otimes G_\ell)
\right)
\]
as a sum over pairings of all the factors. Every pairing containing a
contraction between a factor supported in $I$ and one supported in $J$
vanishes, since $\omega\left(
\xi(\mathbf 1_I\otimes F_r)
\xi(\mathbf 1_J\otimes G_\ell)
\right)=0$. 
Therefore, only pairings internal to the two families of factors contribute. The Wick sum consequently factorizes, and one obtains
\[\omega\left(
	\prod_{r=1}^{n}\xi(\mathbf 1_I\otimes F_r)
	\prod_{\ell=1}^{m}\xi(\mathbf 1_J\otimes G_\ell)
	\right)=
	\omega\left(
	\prod_{r=1}^{n}\xi(\mathbf 1_I\otimes F_r)
	\right)
	\omega\left(
	\prod_{\ell=1}^{m}\xi(\mathbf 1_J\otimes G_\ell)
	\right).
\]
By the definition of independence in 
Equation~\eqref{Eq: Independence of GRV}, the increments of $(B_t)_{t\geq0}$ over disjoint time intervals are independent.
\end{remark}

\begin{remark}
	Observe that, when considering a Grassmannian Brownian motion as per Proposition \ref{Prop: Grassmannian Brownian motion}, Equation \eqref{Eq: Splitting} can be replaced by the decomposition $B_t=(B^1_t,B^2_t)$. Considering the orthogonal splitting $V_x=V^1_x\oplus V^2_x$ where $V^i_x\simeq L^2(\mathbb{R},\mathbb{C}^2)$, $i=1,2$, it holds that, for any $V_x\ni F=(f_1,f_2)$, $F^\prime=(f^\prime_1,f^\prime_2)$,
	$$B_t(F)=B^1_t(f_1)+B^2_t(f_2),$$
	while the two-point correlations are 
	\begin{gather*}
	\omega(B_t^1(f_1)B_s^1(f^\prime_1))=0,\qquad\omega(B_t^2(f_2)B_s^2(f^\prime_2))=0,\\
	\omega(B_t^1(f_1)B_s^2(f^\prime_2))=\min\{t,s\}(\chi f_1,f^\prime_2)_{L^2(\mathbb R;\mathbb C^2)},\quad \omega(B_t^2(f_2)B_s^1(f^\prime_1))=-\min\{t,s\}(\chi f_2,f^\prime_1)_{L^2(\mathbb R;\mathbb C^2)}.	
	\end{gather*}
	
\end{remark}

\noindent  In the following we introduce one auxiliary structure which is necessary for the construction of a Euclidean free Fermionic field. To this end, we recall that the operator $-\partial^2_x+m^2$ with $m^2>0$ and domain $C^\infty_0(\mathbb{R})$ admits a unique self-adjoint extension, which has in addition a strictly positive spectrum. Hence, using functional calculus, we can introduce on $V_x$ the positive, self-adjoint operator $A_x:=(-\partial_x^2+m^2)^{1/2}\otimes\mathbb I_{\mathbb C^4}$. The associated semigroup $e^{-rA_x}$, $r\geq 0$ is bounded on $V_x$ and in addition it preserves $\mathcal{S}(\mathbb R;\mathbb C^4)$. Observe that the construction below uses a white noise constructed out of test functions possibly supported on the full Euclidean time axis, whereas the process $B_t$ denotes only the increment over $[0,t]$.

\begin{definition}\label{Def: Stationary GO-U Field}
We call {\bf stationary Grassmannian Ornstein--Uhlenbeck field} associated with $A_x$ the process $\varphi\equiv(\varphi_t)_{t\in\mathbb{R}}$ such that, for any $t\in\mathbb{R}$ and for any $H\in V_x$,
\begin{equation}\label{Eq: GO-U Field}
\varphi_t(H):=\xi\bigl(s\mapsto \mathbf 1_{(-\infty,t]}(s)e^{-(t-s)A_x}H\bigr).
\end{equation}
\end{definition}

\begin{remark}
We highlight that Equation \eqref{Eq: GO-U Field} is best understood by means of its formal representation
$$\varphi_t(H)=\int_{-\infty}^t\xi_s\bigl(e^{-(t-s)A_x}H\bigr)\,ds.$$
The definition is meaningful also because for every $t\in\mathbb{R}$ and for all $H\in V_x$, the map $s\longmapsto \mathbf{1}_{(-\infty,t]}(s)e^{-(t-s)A_x}H$ belongs to $V\equiv L^2(\mathbb{R};V_x)$ as one can infer observing that, on $(-\infty,t)$ the operator $e^{-(t-s)A_x}$ is bounded with norm controlled by $e^{-m(t-s)}$.
\end{remark}

\begin{remark}
It is worth emphasizing that the process $\varphi$ can also be viewed as the stationary solution of the Grassmannian Ornstein--Uhlenbeck equation. More precisely, for any $H$ lying in the domain of $A_x$ and for any $t\geq 0$, it holds
$$\varphi_t(H)-\varphi_0(H)=-\int_0^t \varphi_s(A_xH)\,ds + B_t(H),$$
where
$$\varphi_0(H)=\xi\bigl(s\mapsto \mathbf 1_{(-\infty,0]}(s)e^{sA_x}H\bigr).$$
At a formal notation level, this translates as
\[
d\varphi_t=-A_x\varphi_t\,dt+dB_t,
\qquad
\varphi_{|t=0}=\varphi_0.
\]
\end{remark}

\noindent In order to make contact with Definition \ref{Def: Bosonic OS Distributions}, we consider $\mathbb{R}^2_+:=[0,+\infty)\times\mathbb{R}$ as well as the test-function space $\mathcal{S}(\mathbb{R}^2_+;\mathbb{C}^4)\equiv\mathcal{S}(\mathbb{R}^2_+)\otimes\mathbb{C}^4$, see Equation \eqref{Eq: Positive time Test Function Space}. For \(F=(f_1,f_2)\in\mathcal S(\mathbb R^2_+;\mathbb C^4)\),
writing \(\varphi=(\varphi^1,\varphi^2)\), we define the smeared field
$$\varphi(F):=\int_0^{+\infty}\varphi_t(F(t,\cdot))\,dt,$$
or, equivalently, 
$$\varphi(F)=\varphi^1(f_1)+\varphi^2(f_2)\in \mathcal{A},\quad\textrm{where}\quad\varphi^i(f_i):=\int_0^{+\infty}
\varphi_t^i(f_i(t,\cdot))\,dt,\; i=1,2.$$

\begin{definition}\label{Def: Euclidean FFF}
	Given a stationary Grassmannian Ornstein--Uhlenbeck field as per Definition \ref{Def: Stationary GO-U Field}, we call {\bf Euclidean Fermionic free field} the operator-valued tempered distribution $\psi=(\psi^1,\psi^2)\in \mathcal{L}(\mathcal{S}(\mathbb{R}^2_+;\mathbb{C}^4),\mathcal{A})$ such that
	$$\psi^1:=\varphi^1,\qquad\psi^2:=D_E^*\varphi^2,\quad\textrm{where}\quad D_E^*=-\gamma_E^\mu\partial_\mu+m.$$
\end{definition}

In order to better grasp the meaning of this definition, using the notation established in Section \ref{Sec: Analytic Structures} and given $F=(f_1,f_2)\in\mathcal{S}(\mathbb{R}^2_+;\mathbb{C}^4)$, one can write
$$\psi^1(f_1)=\int_{\mathbb R} dk\int_0^{+\infty} dt\int_{-\infty}^{t} ds\,\widehat{\xi}^1_i(k,s)e^{-A(k)(t-s)}(\hat{f}_1)_i(-k,t),$$
where the $i$-subscript denotes an implicit sum over the two components in $\mathbb{C}^2$ while $A(k)=\sqrt{k^2+m^2}$.  At the same time, setting $\Gamma_E(k):=\gamma_E^0 A(k)-i\gamma_E^1 k+m\mathbb{I}_2$, where $\mathbb{I}_2$ is the identity matrix acting on $\mathbb{C}^2$, it holds 
$$\psi^2(f_2)=\int_{\mathbb R} dk	\int_0^{+\infty} dt \int_{-\infty}^{t} ds\, \Gamma_E(k)_{ij}\widehat{\xi}^2_i(k,s)e^{-A(k)(t-s)}(\hat{f}_2)_j(-k,t)-	\int_{\mathbb R} dk	\int_0^{+\infty} dt\,(\gamma_E^0)_{ij}	\widehat{\xi}^2_i(k,t)(\hat{f}_2)_j(-k,t),$$
where, also in this case the $i,j$-subscripts denote an implicit sum over components in $\mathbb{C}^2$. Observe that the two-point functions of $\psi$ are determined by the underlying Ornstein--Uhlenbeck covariance. Working at the level of integral kernels,
\begin{gather*}
\omega(\psi^1(\mathrm{x})\psi^1(\mathrm{y}))=0,\qquad\omega(\psi^2(\mathrm{x})\psi^2(\mathrm{y}))=0,\\	
\omega(\psi^1_i(\mathrm{x})\psi^2_j(\mathrm{y}))=(S_E)_{ij}(\mathrm{y}-\mathrm{x}),\quad \omega(\psi^2_i(\mathrm{x})\psi^1_j(\mathrm{y}))=-(S_E)_{ji}(\mathrm{x}-\mathrm{y}),
\end{gather*}
where $S_E$ is as per Equation \eqref{Eq: Sol-fond}, and the componentwise convention fixes the two-point kernel as a function of $\mathrm{y}-\mathrm{x}$.

\subsection{Time-zero fields}\label{Sec: Time-zero fields}

Following the strategy of Nelson~\cite{Nel73}, the reconstruction of a Lorentzian theory requires a notion of fields evaluated at Euclidean time zero. In the present setting this can be formulated in terms of conditional expectations associated with the filtration generated by the white noise, hence bypassing the obstruction highlighted in \cite{OS73a}. Therefore, let $\xi$ be a Grassmannian white noise as in Definition \ref{Def: Grassmannian White Noise}, regarded as an element of $\mathcal{L}(\mathcal S(\mathbb R^2;\mathbb C^4);\mathcal{A})$ and let $B_t$ be the associated Grassmannian Brownian motion as in Proposition \ref{Prop: Grassmannian Brownian motion}. For any  bounded interval $I\subset\mathbb{R}$ and for any $h\in V_x$, we set
$$B_I(h)\assign \xi(\mathbf{1}_I\otimes h).$$
For $t\in\mathbb{R}$, set $V_{\leq t}\assign L^2((-\infty,t];V_x)\subset V$ and define $\mathcal{A}_{\leq t}\assign
\{\xi(F),\xi(F)^*:F\in V_{\leq t}\}^{\prime\prime}\subset\mathcal{A}$, the von Neumann algebra generated by the past fields. Furthermore we assume that there exists a family of conditional expectations, see Definition \ref{Def: Conditional Expectations},
$$\omega_t:\mathcal{A}\longrightarrow\mathcal{A}_{\leq t}\qquad t\in\mathbb{R},$$
such that, for every $F\in V$
$$\omega_t(\xi(F))=\xi(\mathbf{1}_{(-\infty,t]}F).$$

\begin{definition}[Zero-time fields]\label{Def: Zero-Time Fields}
We call {\bf zero-time Fermionic fields} associated to $F=(f_1,f_2)\in\mathcal{S}(\mathbb R^2_+;\mathbb C^4)$ the following elements of $\mathcal{A}$:
$$\psi^1_0(f_1)\assign \omega_0(\psi^1(f_1)),
	\quad
	\psi^2_0(f_2)\assign \omega_0(\psi^2(f_2)),\quad\psi_0(F)\assign \psi^1_0(f_1)+\psi^2_0(f_2).$$
\end{definition}

\begin{remark} In Definition~\ref{Def: Zero-Time Fields} we have assumed the existence of a family of conditional expectations $\omega_t$. This assumption is not too strong as one can infer from the following example. Setting $V_{\leq t}:=P_t[V]$ and $V_{>t}=(\mathbb{I}-P_t)[V]$, where
	$$P_t:V\to V,\quad F\mapsto P_t(F)\assign \mathbf{1}_{(-\infty ,t]}F.$$
We have established the orthogonal, direct sum decomposition $V=V_{\leq t}\oplus V_{>t}$. At the level of Fermionic Fock space, this entails the existence of an isomorphism
$$\mathcal{F}_-(V)\simeq\mathcal{F}_-(V_{\leq t})\widehat\otimes\mathcal{F}_-(V_{>t}),$$
which justifies the identification
$$\mathcal{A}_{\leq t}:=\mathcal{B}(\mathcal{F}_-(V_{\leq t}))\otimes\mathbb{I}=\{\xi(F),\xi(F)^*:F\in V_{\leq t}\}^{\prime\prime}
\subset\mathcal{A}.$$
Bearing in mind  Definition \ref{Def: Conditional Expectations}, $\mathcal{A}_{\leq t}$ plays the role of $\mathcal{A}^\prime$, while, for any $t\in\mathbb{R}$ a conditional expectation is the map
$$\omega_t:\mathcal{A}\to\mathcal{A}_{\leq t},\quad A\mapsto\omega_t(A):=	\bigl(\mathrm{id}\otimes\omega_{>t}\bigr)(A)\otimes\mathbb I,$$
where \(\omega_{>t}\) is the standard vacuum state on $\mathcal{B}(\mathcal{F}_-(V_{>t}))$. Observe that  $\omega_t$ is per construction positive, unital, and it satisfies
$$\omega_t(B_1AB_2)=B_1\omega_t(A)B_2,\qquad A\in\mathcal{A},\quad B_1,B_2\in\mathcal{A}_{\leq t}.$$
Moreover, since $\Omega$, the vacuum vector of $\mathcal{F}_-(V)$, factorises as $\Omega=\Omega_{\leq t}\otimes\Omega_{>t}$, it follows by direct inspection that 
$$\omega(\omega_t(A))=\omega(A),\qquad A\in\mathcal A.$$
This allows to conclude that we have established the existence of a conditional expectation. It is worth also observing that, since $P_t$ commutes with both the real structure $\chi$ and with the covariance $C$, then the decomposition as per Definition \ref{Def: Creator Grassmann Random Variable} 
$$\xi(F)=a(\chi(F))+a^\dagger(C (F))\quad F\in V,$$
yields
$$\omega_t(\xi(F))=\xi(P_t(F))=\xi(\mathbf{1}_{(-\infty,t]}F),$$
coherently with the discussion preceding the definition of the zero-time fields.
\end{remark}

\noindent Working at the level of integral kernels and using the notation introduced in the previous subsection, one obtains that, for $F=(f_1,f_2)\in\mathcal{S}(\mathbb{R}^2_+;\mathbb{C}^4)$
\begin{equation}\label{Eq: Time-zero-fields-kernel}
	\begin{aligned}
		\psi^1_0(f_1)&=\int_{\mathbb{R}}\mathd k \int_0^{+\infty}\mathd t \int_{-\infty}^{0}\mathd s\,\widehat{\xi}^1_i(k,s)
		e^{-A(k)(t-s)}
		(\widehat{f}_1)_i(-k,t),
		\\
		\psi^2_0(f_2)
		&=
		\int_{\mathbb{R}}\mathd k
		\int_0^{+\infty}\mathd t
		\int_{-\infty}^{0}\mathd s\,
		\Gamma_E(k)_{ij}
		\widehat{\xi}^2_i(k,s)
		e^{-A(k)(t-s)}
		(\widehat{f}_2)_j(-k,t),
	\end{aligned}
\end{equation}
where, once more $A(k)=(k^2+m^2)^{1/2}$, while $\Gamma_E(k)=\gamma_E^0A(k)-i\gamma_E^1k+m\mathbb{I}_2$. We can also identify the covariance of the zero-time fields, but, to this end, we need some preliminary notation. Setting $\mathrm{x},\mathrm{y}\in\mathbb R^2_+$ with $\mathrm{x}=(t,x)$ and $\mathrm{y}=(t^\prime,y)$,  we define the reflected Euclidean Dirac kernel by
$$S_E^\vartheta(\mathrm{x},\mathrm{y})\assign S_E(\mathrm{y}-\vartheta(\mathrm{x})),$$
where $\vartheta$ is the Euclidean time-reflection as in Equation \eqref{Eq: Time-Reflection}.

\begin{proposition}\label{Prop: Time-zero covariance}
	For any $F=(f_1,f_2)$ and $H=(h_1,h_2)\in\mathcal{S}(\mathbb R^2_+;\mathbb C^4)$ and given the zero-time Fermionic fields as per Definition \ref{Def: Zero-Time Fields}, it holds 
	\begin{equation}\label{Eq: Two-Point Correlation Zero-Time}
		\omega\bigl(\psi^1_0(f_1)\psi^2_0(h_2)\bigr)
		=\sum_{i,j=1}^2
		\left\langle (S_E^\vartheta)_{ij},
		(f_1)_i\otimes(h_2)_j\right\rangle,
	\end{equation}
	where
	$$(S_E^\vartheta)_{ij}\in\mathcal{S}^\prime\bigl(\mathbb{R}^2_+\times\mathbb{R}^2_+\bigr)$$
	are the components of the reflected Euclidean Dirac kernel. At the same time, 
	$$\omega\bigl(\psi^1_0(f_1)\psi^1_0(h_1)\bigr)=0,\qquad\omega\bigl(\psi^2_0(f_2)\psi^2_0(h_2)\bigr)=0.$$
\end{proposition}

\begin{proof}
	Combining \eqref{Eq: Time-zero-fields-kernel} and the two-point correlation functions in Remark \ref{Rem: Algebra Valued White Noise}, one obtains
	$$\omega\bigl(\psi^1_0(f_1)\psi^2_0(h_2)\bigr)=\int_{\mathbb{R}}\mathd k\int_0^{+\infty}\mathd t\int_0^{+\infty}\mathd t^\prime
		\int_{-\infty}^{0}\mathd s\,e^{-A(k)(t-s)}e^{-A(k)(t'-s)}\Gamma_E(-k)_{ij}(\widehat{f}_1)_i(-k,t)(\widehat{h}_2)_j(k,t').$$
Taking the integral over $s$ yields
$$\omega\bigl(\psi^1_0(f_1)\psi^2_0(h_2)\bigr)=\int_{\mathbb{R}}\mathd k\int_0^{+\infty}\mathd t\int_0^{+\infty}\mathd t^\prime\, \frac{\Gamma_E(-k)_{ij}}{2A(k)} e^{-A(k)(t+t')}(\widehat{f}_1)_i(-k,t)(\widehat{h}_2)_j(k,t').$$
This is the spatial Fourier representation of the reflected kernel $S_E^\vartheta(\mathrm{x},\mathrm{y})=S_E(\mathrm{y}-\vartheta(\mathrm{x}))$ appearing on the right-hand side of Equation \eqref{Eq: Two-Point Correlation Zero-Time}. The remaining identity is instead a by-product of Remark \ref{Rem: Algebra Valued White Noise}.
\end{proof}

To conclude this subsection we establish an interplay between time-zero fields and Wick polynomials. We recall that the latter are maps $W^\psi_n\in\mathcal{L}(\mathcal{S}(\mathbb{R}^2_+;\mathbb{C}^4)^{\widehat{\otimes}n};\mathcal{A})$, $n\geq 1$, recursively defined as follows: Given $F_1,\cdots,F_n\in\mathcal S(\mathbb R^2_+;\mathbb C^4)$, 
\begin{equation}\label{Eq: Wick Polynomials order 1 and 2}
W^\psi_1(F_1)\assign \psi(F_1),\quad W^\psi_2(F_1\otimes F_2)\assign\psi(F_1)\psi(F_2)-\omega(\psi(F_1)\psi(F_2))\mathbb{I}_{\mathcal{A}},
\end{equation}
where $\omega$ is the state of the underlying non-commutative probability space, while $\mathbb{I}_{\mathcal{A}}$ is the unit element of the underlying algebra. For ease of notation, we shall omit henceforth to write explicitly the identity element leaving it implicit. For $n\geq 3$, the Wick rule yields
\begin{equation}\label{Eq: Wick Polynomials order 3}
W^\psi_n(F_1\otimes\cdots\otimes F_n)\assign\psi(F_1)W^\psi_{n-1}(F_2\otimes\cdots\otimes F_n)-\sum_{j=2}^n(-1)^j	\omega(\psi(F_1)\psi(F_j))	W^\psi_{n-2}(F_2\otimes\cdots\widetilde{F_j}\cdots\otimes F_n),
\end{equation}
where the tilde denotes an omitted factor. For simplicity and to make contact with the notation often used in the theoretical physics literature, we shall also write
\begin{equation}\label{Eq: Wick Polynomials vs :}
\Wick{\psi(F_1)\cdots\psi(F_n)}\assign W^\psi_n(F_1\otimes\cdots\otimes F_n).
\end{equation}

\begin{remark}\label{Rem: Wick-Polynomials Zero-Time Fields}
	In the following we will also consider Wick polynomials constructed out of zero-time Fermionic fields. Similarly, they can be defined recursively setting 
	\begin{equation}\label{Eq: Wick Polynomials order 1 and 2 - zero time}
	W^{\psi_0}_1(F_1)=\Wick{\psi_0(F_1)}\assign \psi_0(F_1),\quad W^{\psi_0}_2(F_1\otimes F_2)=\Wick{\psi_0(F_1)\psi_0(F_2)}\assign\psi_0(F_1)\psi_0(F_2)-\omega(\psi_0(F_1)\psi_0(F_2)),
	\end{equation}
	while, for $n\geq 3$, the Wick rule once more yields
	\begin{gather}
	W^{\psi_0}_n(F_1\otimes\cdots\otimes F_n)=\Wick{\psi_0(F_1)\cdots\psi_0(F_n)}\notag\\
	\assign\psi_0(F_1)\Wick{\psi_0(F_2)\cdots\psi_0(F_n)}-\sum_{j=2}^n(-1)^j	\omega(\psi_0(F_1)\psi_0(F_j))\Wick{\psi_0(F_2)\cdots\widetilde{\psi_0(F_j)}\dots\psi_0(F_n)}.\label{Eq: Wick Polynomials vs : - zero time}
	\end{gather}
\end{remark}

\noindent We can now establish the sought interplay between Wick ordering and time-zero fields. To favor a smooth reading of the main body of this manuscript, we postpone the proof of the following proposition to Appendix \ref{Appendix}. Observe that $\omega_0$ is obtained from the conditional expectations $\omega_t$ setting $t=0$.

\begin{proposition}\label{Prop: Wick_vs_time_zero}
	For any $n\in\mathbb{N}$ and for any $F_1,\dots F_n\in\mathcal S(\mathbb R^2_+;\mathbb C^4)$, 
$$\omega_0\bigl(W^\psi_n(F_1\otimes\cdots\otimes F_n)\bigr)=W^{\psi_0}_n(F_1\otimes\cdots\otimes F_n).$$
\end{proposition}

\begin{remark}\label{Rem: Markovian time-zero expectation}
In Nelson's Markov-field approach to Bosonic quantization \cite{Nel73}, usually one considers the conditional expectation onto the time-zero algebra rather than onto the larger algebra generated by the entire past. Indeed, for observables supported at nonnegative times, the Markov property implies that conditioning on the entire past is equivalent to conditioning on the field at time zero. An analogous statement holds in the Fermionic framework presented here. Let
$$R_0H(s)\assign\mathbf{1}_{(-\infty,0]}(s)e^{sA_x}H,\qquad\mathcal{K}_0\assign\overline{R_0(V_x)}^{\,V_{\leq 0}},$$
and set $\mathcal{A}_0\assign\{\xi(v),\xi(v)^*:v\in\mathcal K_0\}^{\prime\prime}\subset\mathcal{A}_{\leq 0}$. The orthogonal decomposition
$$V_{\leq 0}=\mathcal{K}_0\oplus\mathcal{K}_0^\perp$$
defines a vacuum conditional expectation
$$E_0:\mathcal{A}_{\leq 0}\longrightarrow\mathcal{A}_0.$$
Consequently,
$$\widetilde{\omega}_0\assign E_0\circ\omega_0$$
is a conditional expectation onto the time-zero algebra. Since
$$W_n^{\psi_0}(F_1\otimes\cdots\otimes F_n)\in\mathcal{A}_0,$$
Proposition~\ref{Prop: Wick_vs_time_zero} gives
$$\widetilde{\omega}_0\bigl(W_n^\psi(F_1\otimes\cdots\otimes F_n)\bigr)=\omega_0	\bigl(W_n^\psi(F_1\otimes\cdots\otimes F_n)\bigr)=W_n^{\psi_0}(F_1\otimes\cdots\otimes F_n).$$
Thus, on the Wick polynomials considered here, conditioning on the entire past is equivalent to conditioning on the time-zero algebra. The two conditional expectations do not, however, coincide on arbitrary elements of $\mathcal{A}_{\leq 0}$.
\end{remark}

\subsection{The physical one-particle space at time zero}\label{Sec: Physical one-particle space}
In this subsection we show how the field $\psi_0=(\psi^1_0,\psi^2_0)$, that is the zero-time Fermionic field in Definition \ref{Def: Zero-Time Fields}, can be used to build a version of the time zero physical space of Fermionic fields. In comparison with \cite{Nel73}, as a last step in our analysis of Dirac fields, we need to establish an Euclidean time evolution of the zero-time fields. In the following we shall consider all Wick polynomials as per Remark \ref{Rem: Wick-Polynomials Zero-Time Fields} and we shall use all notations and conventions established in Section \ref{Sec: Time-zero fields}. Given $F=(f_1,f_2)\in\mathcal{S}(\mathbb{R}^2_+;\mathbb{C}^4)$ we define
\begin{equation}\label{Eq: Q1 and Q2}
	\widehat{Q_1(f_1)}(k)\assign \int_0^{+\infty}e^{-A(k)t}\widehat{f}_1(k,t)\,\mathd t\quad\textrm{and}\quad\widehat{Q_2(f_2)}(k)\assign\int_0^{+\infty}e^{-A(k)t}\Gamma_E(-k)\widehat{f}_2(k,t)\,\mathd t.
\end{equation}
Observe that $Q_1(f_1),Q_2(f_2)\in \mathcal{S}(\mathbb R;\mathbb C^2)$ and in the last equation, we use $\Gamma_E(-k)$ for consistency with respect to Equation \eqref{Eq: Time-zero-fields-kernel} where the test function appears as $\widehat{f}_2(-k,t)$. We set 
\begin{equation}\label{Eq: Full Q}
	\mathcal{S}(\mathbb{R}^2_+;\mathbb{C}^4)\ni F\mapsto Q(F)\assign (Q_1(f_1),Q_2(f_2))\qquad \mathcal{E}_0\assign Q\bigl(\mathcal{S}(\mathbb{R}^2_+;\mathbb{C}^4)\bigr)\subset\mathcal{S}(\mathbb{R};\mathbb{C}^4).
\end{equation}

It is important to note that the inclusion defining $\mathcal{E}_0$ is strict as one can infer from Equation \eqref{Eq: Q1 and Q2} since the matrix $\Gamma_E(-k)$ has vanishing determinant for every $k\in\mathbb{R}$. Hence $Q_2$ is not surjective on the whole $\mathcal{S}(\mathbb{R};\mathbb{C}^2)$ and, therefore, one is forced to work with $\mathrm{Ran}(Q)$. In order to better understand the construction we define the map
$$\Psi_0:\mathcal{E}_0\to\mathcal{A},\quad Q(F)\mapsto\Psi_0(Q(F))\assign\psi_0(F).$$

\begin{proposition}\label{Prop: One-particle norm}
For every $F,G\in\mathcal{S}(\mathbb{R}^2_+;\mathbb{C}^4)$,
$$\left\langle\psi_0(F)\Omega,\psi_0(G)\Omega\right\rangle_{\mathcal{F}_-(V)}=\left(Q(F),\frac{1}{2A_x}Q(G)		\right)_{V_x}.$$
Consequently $\ker\psi_0=\ker Q$, and the vacuum-vector map $q\mapsto\Psi_0(q)\Omega$ is an isometry when $\mathcal{E}_0$ is endowed with the inner product
\begin{equation}\label{eq:scalarproduct}
\langle q,r\rangle_0\assign\left(q,\frac{1}{2A_x}r\right)_{V_x},\quad q,r \in \mathcal{E}_0
\end{equation}
\end{proposition}
\begin{proof}
For $q\in\mathcal E_0$, set
\begin{equation}\label{Eq: Aux1}
(Rq)(s)\assign\mathbf{1}_{(-\infty,0]}(s)e^{sA_x}q.
\end{equation}
Equation \eqref{Eq: Time-zero-fields-kernel} gives $\psi_0(F)=\xi(RQ(F))$. Hence, using 		$\xi(v)\Omega=a^\dagger(Cv)\Omega$ and $C^*C=\mathbb I$,
$$\begin{aligned}
\left\langle\psi_0(F)\Omega,\psi_0(G)\Omega\right\rangle&=(RQ(F),RQ(G))_V\\
&=\int_{-\infty}^0\left(e^{sA_x}Q(F),e^{sA_x}Q(G)\right)_{V_x}\mathd s\\
&=\left(Q(F),\frac{1}{2A_x}Q(G)\right)_{V_x}.
\end{aligned}
$$
Since $A_x\geq m>0$, the last quadratic form vanishes precisely when $Q(F)=0$, proving $\ker\psi_0=\ker Q$.
\end{proof}

\noindent Thus since $\ker(Q)=\ker\psi_0$ and $Q$ is continuous, we endow $\mathcal{E}_0$ with the standard quotient topology, upon the identification $\mathcal{E}_0\simeq\mathcal{S}(\mathbb{R}^2_+;\mathbb{C}^4)/\ker\psi_0$. We denote by 
\begin{equation}\label{eq:H01}
\mathcal{H}_0^{(1)}\assign \overline{\mathcal{E}_0}^{\|\cdot\|_0}
\end{equation} 
the completion of the space of functions $\mathcal{E}_0$ with respect to the scalar product $\langle \cdot, \cdot \rangle_0$. On account of Proposition \ref{Prop: One-particle norm}, the map $\Psi_0$ is injective, and if we write 
\begin{equation}\label{eq:J0}
J_0:\mathcal{E}_0 \rightarrow V\subset \mathcal{F}_-(V),\quad J_0(q)=\Psi_0(q)\Omega.
\end{equation} 
this is an isometry. Hence, being a bounded operator, there exists a unique continuous extension of $J_0$ from $\mathcal{E}_0$ to the whole $\mathcal{H}_0^{(1)}$. The role of $\Psi_0$ is that of a reparameterization of the zero-time Fermionic fields which is the natural starting point to construct a Hamiltonian evolution.\\

\begin{remark}
Unlike $J_0$, the map $\Psi_0:\mathcal{E}_0 \rightarrow \mathcal{A}$ is not an isometry, at least when $\mathcal{A}$ is equipped with its standard $C^*$-norm. At the same time, the vacuum state identifies a seminorm on \(\mathcal A\) by
$$\|A\|_{2,\omega}^2\assign\omega(A^*A)=\|A\Omega\|^2,$$
which has null space $\mathcal{N}_\omega=\{A\in\mathcal A:A\Omega=0\}$. Proposition~\ref{Prop: One-particle norm} gives
$\|\Psi_0(q)\|_{2,\omega}^2=\|q\|_0^2$. Thus $J_0$ identifies $\mathcal{E}_0$ isometrically with a subspace of the GNS $L^2$-completion of $\mathcal{A}/\mathcal{N}_\omega$. The vacuum state is not faithful on $\mathcal{A}$, so this is not a norm on the operators themselves. More in general, it is possible to realize the Fermionic expectation with a faithful quasi-free state and the non-commutative $L^p$-framework, see, {\it e.g.}, \cite{DFGG25,AM82}.
\end{remark}

\noindent  Unfortunately, although $\mathcal{E}_0$ is very useful as a starting point to build the physical Hilbert space, it is too large. In order to understand the reason it is convenient to introduce the following operators: Assume $m>0$ and set
\[
h_D(k)\assign i\gamma_E^0\gamma_E^1k+m\gamma_E^0,
\qquad
P_\pm(k)\assign
\frac12\left(\mathbb I_2\pm\frac{h_D(k)}{A(k)}\right).
\]
It holds that $h_D(k)^*=h_D(k)$, $h_D(k)^2=A(k)^2\mathbb I_2$, while $P_\pm(k)$ are complementary orthogonal projections of rank one. These functions define continuous multipliers on	$\mathcal S(\mathbb R;\mathbb C^2)$ and the associated operators, which we denote again with the same symbols, commute with the operator $A_x$. We have that 
	\begin{equation}\label{Eq: Gamma energy projection}
	\gamma_E^0\Gamma_E(-k)=2A(k)P_+(k),
	\qquad
	\Gamma_E(-k)=2A(k)\gamma_E^0P_+(k).
\end{equation} 

\begin{proposition}\label{Prop: Structure auxiliary one-particle space} 
Under the given assumptions, it holds that 
\begin{equation}\label{Eq: Exact range Q}
\operatorname{Ran}Q_1=\mathcal{S}(\mathbb{R};\mathbb{C}^2),\qquad\operatorname{Ran}Q_2= \gamma_E^0P_+\mathcal{S}(\mathbb{R};\mathbb{C}^2),
\end{equation}
and hence
\begin{equation}\label{Eq: Exact auxiliary space}
\mathcal{E}_0=\mathcal{S}(\mathbb{R};\mathbb{C}^2)\oplus\gamma_E^0P_+\mathcal{S}(\mathbb{R};\mathbb{C}^2).
\end{equation}
\end{proposition}
 		
\begin{proof}
To determine the ranges, choose $\rho\in C_c^\infty((0,\infty))$ with $\int_0^\infty\rho(s)\,\mathd s=1$. For $q\in\mathcal{S}(\mathbb{R};\mathbb{C}^2)$, set
$$\widehat{f}_q(k,t)\assign(A(k)+1)e^{A(k)t}\rho((A(k)+1)t)\widehat q(k).$$
Since $A(k)t$ is bounded on the support of $\rho((A(k)+1)t)$, we have $f_q\in\mathcal{S}(\mathbb{R}^2_+;\mathbb{C}^2)$, and
$$\widehat{Q_1(f_q)}(k)=\widehat{q}(k)\int_0^\infty	(A(k)+1)\rho((A(k)+1)t)\,\mathd t=\widehat{q}(k).$$
Thus $Q_1$ is surjective. Applying the same argument componentwise to the Laplace transform occurring in $Q_2$, and using
\eqref{Eq: Gamma energy projection}, yields
$$\operatorname{Ran}Q_2=\Gamma_E(-\cdot)\mathcal{S}(\mathbb{R};\mathbb{C}^2)=2A_x\gamma_E^0P_+\mathcal{S}(\mathbb{R};\mathbb{C}^2)=\gamma_E^0P_+\mathcal{S}(\mathbb{R};\mathbb{C}^2),$$
because both $A_x$ and $A_x^{-1}$ preserve the Schwartz space.	Equations \eqref{Eq: Exact range Q} and \eqref{Eq: Exact auxiliary space} follow.
\end{proof}
 		 
\noindent  Starting from Equation \eqref{Eq: Exact auxiliary space}, we can endow $\mathcal{S}(\mathbb{R};\mathbb{C}^2)\oplus\mathcal{S}(\mathbb{R};\mathbb{C}^2)$ with the inner product $\langle\cdot,\cdot\rangle_0$ as per Equation \eqref{eq:scalarproduct}. Since, for $q\in\mathcal{S}(\mathbb{R};\mathbb{C}^4)$,
\begin{equation}\label{Eq: Aux4}
\|q\|^2_0=\int\limits_{\mathbb{R}}\frac{\|\hat{q}(k)\|^2_{\mathbb{C}^4}}{2\sqrt{k^2+m^2}}dk,
\end{equation}
and since $m>0$, this norm is equivalent to that on $H^{-\frac{1}{2}}(\mathbb{R};\mathbb{C}^4)$. Hence, upon completion, we end up with $H^{-1/2}(\mathbb{R}, \mathbb{C}^2) \oplus H^{-1/2}(\mathbb{R}, \mathbb{C}^2)$. Observing that $\mathcal{E}_0=\left(\mathbb{I}\oplus \gamma_E^0 P_+\right)[\mathcal{S}(\mathbb{R};\mathbb{C}^2)\oplus\mathcal{S}(\mathbb{R};\mathbb{C}^2)]$ and that $\mathbb{I}\oplus \gamma_E^0 P_+$ is a bounded operator with closed range, being $P_+$ an orthogonal projector, we can conclude that the completion of $\mathcal{E}_0$ with respect to $\langle\cdot,\cdot\rangle_0$ as per Equation \eqref{eq:H01} is
\begin{equation}\label{Eq: Aux2}
\mathcal{H}_0^{(1)}=H^{-1/2}(\mathbb{R};\mathbb{C}^2) \oplus \gamma_E^0P_+H^{-1/2}(\mathbb{R}, \mathbb{C}^2).
\end{equation}
In view of Equation \eqref{eq:J0}, the map $J_0$ identifies $\mathcal{E}_0$ isometrically with a subspace of the one-particle Hilbert subspace of $\mathcal{F}_-(V)$. Hence, we can refer to $\mathcal{E}_0$ and, more properly to $\mathcal{H}_0^{(1)}$ as auxiliary time-zero one-particle spaces. Since $\gamma_E^0$ is invertible and $P_+$ is a rank $1$ projection, it turns that, on each fiber, $\mathcal{H}_0^{(1)}$ has three independent components, while a Dirac spinor in two dimensions has only two. In order to establish which is the redundant one, we use the following criterion. Let $f\in\mathcal{S}(\mathbb{R};\mathbb{C}^2)$ and denote by $q_f$ the corresponding element $(f,0)\in\mathcal{E}_0$. Considering $(Rq_f)(t)$ as per Equation \eqref{Eq: Aux1} and writing the Dirac operator as  $D_E=\gamma_E^0(\partial_t+h_D(k))$ in view of Equation \eqref{Eq:dirac-op}, it holds that $(D_E\oplus D^*_E)(Rq_f)=0$ if and only if 
$$\left(A(k)+h_D(k)\right)\hat{f}(k)=2A(k)P_+(k)\hat{f}(k)=0.$$
Since $A(k)>0$, hence it is invertible, this entails that $\hat{f}(k)$ must lie in $\textrm{Ran}(P_-)$, the orthogonal complement of $\textrm{Ran}(P_+)$, namely $f\in P_-[\mathcal{S}(\mathbb{R};\mathbb{C}^2)]$. Observe that, should we have started from a generic element of the form $(0,f^\prime)$ with $f^\prime\in\mathcal{S}(\mathbb{R};\mathbb{C}^2)$ requiring that $e^{t A_x}f^\prime$ satisfies the adjoint Dirac equation for $t\leq 0$, we would have drawn the conclusion that $f^\prime\in \gamma_E^0P_+[\mathcal{S}(\mathbb{R};\mathbb{C}^2)]$. This recovers the restriction already imposed and it justifies once more the factor $\gamma_E^0P_+H^{-1/2}(\mathbb{R}, \mathbb{C}^2)$ in Equation \eqref{Eq: Aux2}. Furthermore, this also motivates introducing the projection onto the physical degrees of freedom:
\begin{equation}\label{Eq: Physical projection}
\mathbb{P}_{\mathrm{phys}}:\mathcal E_0\longrightarrow\mathcal E_0,\quad (q_1,q_2)\mapsto\mathbb{P}_{\mathrm{phys}}(q_1,q_2)\assign(P_-q_1,q_2).
\end{equation}
Having identified $\mathbb{P}_{\mathrm{phys}}[\mathcal{E}_0]$ as the admissible time-zero data, we should ascertain its compatibility with the auxiliary one-particle structure. In particular, we show that it is orthogonal with respect to \(\langle\cdot,\cdot\rangle_0\) and commutes with \(A_x\). Consequently, it extends to the Hilbert completion and its range defines a closed \(A_x\)-invariant subspace.
\begin{lemma}\label{Prop: Physical one-particle space} 
Given $\mathbb{P}_{\mathrm{phys}}$ as in Equation \eqref{Eq: Physical projection}, it holds that $[\mathbb P_{\mathrm{phys}}, A_x]=0$ and $\mathbb{P}_{\mathrm{phys}}$ is an orthogonal projection with respect to the inner product $\langle \cdot,\cdot\rangle_0$ as in Equation \eqref{eq:scalarproduct}. Hence it extends uniquely to an orthogonal projection on $\mathcal{H}^{(1)}_0$ and its range and kernel are
\begin{subequations}
\begin{align}\label{Eq: Physical one-particle core}
	&\mathcal{E}_{0,\mathrm{phys}}\assign\mathbb P_{\mathrm{phys}}[\mathcal{E}_0]=P_-\mathcal S(\mathbb R;\mathbb C^2)\oplus\gamma_E^0P_+\mathcal S(\mathbb R;\mathbb C^2)&&\\
	\label{Eq: Physical one-particle core2}
	&\mathcal{H}_{\phys}^{(1)}\assign\mathbb P_{\phys}[\mathcal{H}_0^{(1)}]=P_-H^{-1/2}(\mathbb R;\mathbb C^2)\oplus\gamma_E^0P_+H^{-1/2}(\mathbb R;\mathbb C^2)&&\\
	\label{Eq: Kernel physical projection}
	&\ker\mathbb{P}_{\mathrm{phys}}=P_+\mathcal{S}(\mathbb{R};\mathbb{C}^2)\oplus\{0\}.&&
\end{align}
\end{subequations}
\end{lemma}
	
\begin{proof}
The identities $P_-^2=P_-$ and $[P_-,A_x]=0$ imply respectively that $\mathbb P_{\mathrm{phys}}^2=\mathbb P_{\mathrm{phys}}$ and $[\mathbb P_{\mathrm{phys}},A_x]=0$. Since $P_-^*=P_-$ and since $(2A_x)^{-1}$ commutes with $P_-$,
$$\langle\mathbb{P}_{\mathrm{phys}}q,r\rangle_0=\langle q,\mathbb{P}_{\mathrm{phys}}r\rangle_0, \quad q,r\in\mathcal E_0.$$
Thus $\mathbb P_{\mathrm{phys}}$ is an orthogonal projector and, being bounded, it admits a unique extension to $\mathcal{H}^{(1)}_0$. Equations \eqref{Eq: Physical one-particle core}, \eqref{Eq: Physical one-particle core2} and \eqref{Eq: Kernel physical projection} follow immediately from Equation \eqref{Eq: Exact auxiliary space}. 
\end{proof}

\noindent We conclude this subsection highlighting a structural feature related to $\mathbb{P}_{\mathrm{phys}}$, namely, that the physically admissible time-zero data can be obtained starting from the original space of test functions $\mathcal{S}(\mathbb{R}^2_+;\mathbb{C}^4)$ as follows. Recalling Equation \eqref{Eq: Full Q}, if we set $Q_{\mathrm{phys}}\assign\mathbb{P}_{\mathrm{phys}}Q=\bigl(P_-Q_1,Q_2\bigr)$, then
\begin{equation}\label{Eq: Physical one-particle null space}
\ker Q_{\mathrm{phys}}=\left\{F=(f_1,f_2): P_-Q_1(f_1)=0,\quad Q_2(f_2)=0\right\}.
\end{equation}
Consequently, being $Q_{\mathrm{phys}}$ surjective onto $\mathcal E_{0,\mathrm{phys}}$, the first isomorphism theorem entails that
\begin{equation}\label{Eq: Physical quotient}
\frac{\mathcal{S}(\mathbb R^2_+;\mathbb C^4)}{\ker Q_{\mathrm{phys}}}\simeq\mathcal E_{0,\mathrm{phys}}.
\end{equation}

\subsection{Fock spaces and Wick vectors}

Having identified the auxiliary and physical one-particle spaces, we investigate the corresponding Fermionic many-particle sectors. Our first goal is to show that the Wick polynomials constructed from the time-zero fields reproduce, when applied to the vacuum, the antisymmetric tensor powers of the one-particle space. This establishes an isometric identification between the Wick sectors and the corresponding Fermionic Fock sectors and it provides the Hilbert space framework used in the construction of the time evolution. To this end, recalling Equation \eqref{Eq: Exact auxiliary space}, for $n\geq1$, set
\begin{equation}\label{Eq: Auxiliary n-particle sector}
\mathcal E_0^{\wedge n}\assign\bigwedge^n\mathcal E_0,\qquad\mathcal H_0^{(n)}\assign\overline{\mathcal E_0^{\wedge n}}=\bigwedge^n\mathcal H_0^{(1)},
\end{equation}
with the scalar product induced by $\langle\cdot,\cdot\rangle_0$ and define 
$$\begin{aligned}
W_{n,0}^{\Psi_0}:\mathcal{E}_0^{\wedge n}&\longrightarrow\mathcal A_{\leq0}\\
q_1\wedge\cdots\wedge q_n&\longmapsto\Wick{\Psi_0(q_1)\cdots\Psi_0(q_n)}
\end{aligned}
\qquad\mathcal W_{n,0}^{\Psi_0}(u)\assign W_{n,0}^{\Psi_0}(u)\Omega.$$
	
\begin{proposition}\label{Prop: Auxiliary Wick sectors}
For $n\geq1$ and $q_1,\ldots,q_n\in\mathcal E_0$,
\begin{equation}\label{Eq: Wick vector exterior product}
\mathcal W_{n,0}^{\Psi_0}(q_1\wedge\cdots\wedge q_n)=J_0q_1\wedge\cdots\wedge J_0q_n.
\end{equation}
Consequently $\mathcal W_{n,0}^{\Psi_0}$ extends to an isometric embedding of $\mathcal{H}_0^{(n)}$ into $\mathcal{F}_-(V)$, and $W_{n,0}^{\Psi_0}$ is injective on $\mathcal{E}_0^{\wedge n}$.
\end{proposition}
	
\begin{proof}
For $q\in\mathcal E_0$, using Equation \eqref{Eq: Aux1}, $(Rq)(s)\assign\mathbf 1_{(-\infty,0]}(s)e^{sA_x}q$. Proposition~\ref{Prop: One-particle norm} gives
$$\Psi_0(q)=\xi(Rq)=a^\dagger(J_0q)+a(\chi Rq)\quad\textrm{and}\quad\omega\bigl(\Psi_0(q)\Psi_0(r)\bigr)=(\chi Rq,J_0r)_V.$$
In the recursive Wick rule, the contraction terms cancel the action	of $a(\chi Rq_1)$ on $J_0q_2\wedge\cdots\wedge J_0q_n$. More precisely its contraction with the $j$-th factor has coefficient $(-1)^j\omega(\Psi_0(q_1)\Psi_0(q_j))$ and only $a^\dagger(J_0q_1)$ remains. Induction on $n$, starting from $\Psi_0(q)\Omega=J_0q$, yields Equation \eqref{Eq: Wick vector exterior product}. Since $J_0$ is an isometry,
$$\left\langle\mathcal{W}_{n,0}^{\Psi_0}(q_1\wedge\cdots\wedge q_n),\mathcal{W}_{n,0}^{\Psi_0}(r_1\wedge\cdots\wedge r_n)\right\rangle_{\mathcal{F}_-(V)}=\det\bigl(\langle q_i,r_j\rangle_0\bigr)_{i,j=1}^n.$$
This is the scalar product on $\mathcal E_0^{\wedge n}$ and it gives the stated extension. To conclude, observe that, given $u\in\mathcal{E}_0^{\wedge n}$, $W_{n,0}^{\Psi_0}(u)=0$ implies $\mathcal{W}_{n,0}^{\Psi_0}(u)=0$, hence $u=0$. This establishes injectivity.
\end{proof}

We now restrict the preceding construction to the physical one-particle subspace identified in Section \ref{Sec: Physical one-particle space}. Starting from Equation \eqref{Eq: Physical one-particle core}, we set
\begin{equation}\label{Eq: Physical n-particle sector}
\mathcal{E}_{0,\phys}^{\wedge n}\assign\bigwedge^n\mathcal{E}_{0,\phys},\qquad\mathcal{H}_{\phys}^{(n)}\assign\overline{\mathcal{E}_{0,\phys}^{\wedge n}}=\bigwedge^n\mathcal{H}_{\phys}^{(1)},
\end{equation}
and let $J_{0,\phys}\assign J_0|_{\mathcal{H}_{\phys}^{(1)}}$.
	
\begin{corollary}\label{Prop: Physical Wick sectors}
The restriction
$$\mathcal{W}_{n,\phys}^{\Psi_0}\assign\mathcal{W}_{n,0}^{\Psi_0}|_{\mathcal{H}_{\phys}^{(n)}}=J_{0,\phys}^{\wedge n}:\mathcal{H}_{\phys}^{(n)}\longrightarrow\mathcal{F}_-(V),$$
is an isometric embedding. For $q_j\in\mathcal{E}_{0,\phys}$ it satisfies
$$\mathcal{W}_{n,\phys}^{\Psi_0}(q_1\wedge\cdots\wedge q_n)=W_{n,0}^{\Psi_0}(q_1\wedge \ldots \wedge q_n)\Omega.$$
In particular, $W_{n,0}^{\Psi_0}$ is injective on $\mathcal{E}_{0,\phys}^{\wedge n}$.
\end{corollary}
	
\begin{proof}
It suffices to apply Proposition~\ref{Prop: Auxiliary Wick sectors} to	$\bigwedge^n\mathcal{H}_{\phys}^{(1)}$.
\end{proof}
	
\noindent To include the vacuum sectors, we set
$$\mathcal{H}_0^{(0)}=\mathcal{H}_{\phys}^{(0)}\assign\mathbb{C},\qquad\mathcal{W}_{0,0}^{\Psi_0}(z)=\mathcal{W}_{0,\phys}^{\Psi_0}(z)\assign z\Omega,$$
so that the auxiliary and physical Fock spaces are
\begin{equation}\label{Eq: Auxiliary and physical Fock spaces}
\mathcal{F}_0\assign\bigoplus_{n=0}^{\infty}\mathcal{H}_0^{(n)}=\mathcal{F}_-\bigl(\mathcal{H}_0^{(1)}\bigr),\qquad\mathcal{F}_{\phys}\assign\bigoplus_{n=0}^{\infty}\mathcal{H}_{\phys}^{(n)}=\mathcal{F}_-\bigl(\mathcal{H}_{\phys}^{(1)}\bigr).
\end{equation}
The preceding propositions also establish the existence of the isometric embeddings
$$\Gamma_-(J_0)=\bigoplus_{n=0}^{\infty}\mathcal{W}_{n,0}^{\Psi_0}:\mathcal{F}_0\longrightarrow\mathcal{F}_-(V),\qquad\Gamma_-(J_{0,\phys})=\bigoplus_{n=0}^{\infty}\mathcal{W}_{n,\phys}^{\Psi_0}:\mathcal{F}_{\phys}\longrightarrow\mathcal{F}_-(V).$$

\subsection{Time Evolution and Hamiltonian}\label{Sec: Time Evolution and Hamiltonian}

We now turn our attention to the time evolution of the theory. First of all we show that Euclidean time translations induce a strongly continuous contraction semigroup on the auxiliary and physical one-particle spaces, identifying its positive generator. Subsequently we use the Markovian structure, through the conditional expectation onto the past algebra, to lift this evolution to Wick polynomials and to the corresponding Fock spaces. The resulting Fock Hamiltonian will be shown to coincide with the second quantization of the one-particle counterpart.

\subsubsection{The one-particle semigroup}

For \(\tau\geq0\), define the Euclidean time translation
$$\Theta_\tau:\mathcal{S}(\mathbb{R}^2_+;\mathbb{C}^4)\longrightarrow\mathcal{S}(\mathbb{R}^2_+;\mathbb{C}^4),\qquad(\Theta_\tau F)(t,x)\assign F(t-\tau,x).$$
By direct inspection of Equation~\eqref{Eq: Q1 and Q2}, for every $F=(f_1,f_2)\in\mathcal{S}(\mathbb{R}^2_+;\mathbb{C}^4)$,
$$\widehat{Q_1(\Theta_\tau f_1)}(k)=e^{-\tau A(k)}\widehat{Q_1(f_1)}(k),\qquad\widehat{Q_2(\Theta_\tau f_2)}(k)=e^{-\tau A(k)}\widehat{Q_2(f_2)}(k).$$
Consequently,
$$Q(\Theta_\tau F)=e^{-\tau A_x}Q(F),$$
and hence the following map is well defined:
\begin{equation}\label{Eq: T_tau}
T_\tau:\mathcal{E}_0\longrightarrow\mathcal{E}_0,\qquad	T_\tau Q(F)\assign Q(\Theta_\tau F)=e^{-\tau A_x}Q(F).
\end{equation}

\begin{proposition}\label{Prop: 1-Particle Hamiltonian}
The family $(T_\tau)_{\tau\geq 0}$ as in Equation \eqref{Eq: T_tau} extends uniquely to a strongly continuous self-adjoint contraction semigroup on $\mathcal{H}_0^{(1)}$, see Equation \eqref{eq:H01}. Its positive Hamiltonian is 
$$H_{1,0}\assign A_x,\qquad D(H_{1,0})=\{q\in\mathcal{H}_0^{(1)}:A_xq\in\mathcal{H}_0^{(1)}\},$$
so that $T_\tau=e^{-\tau H_{1,0}}$ and its infinitesimal generator is $-H_{1,0}$. In particular $\mathcal{E}_0$ is a core for $H_{1,0}$ and for every $q\in\mathcal{E}_0$,
$$\left.\frac{\mathd}{\mathd\tau}T_\tau q\right|_{\tau=0}=-A_xq.$$
\end{proposition}

\begin{proof}
On $\mathcal{E}_0$, the semigroup property follows immediately from $T_\tau$ being of the form $e^{-\tau A_x}$. Moreover,
\begin{equation}\label{Eq: Aux3}
\|T_\tau q\|_0^2=\int_{\mathbb R}\frac{e^{-2\tau A(k)}}{2A(k)}\|\hat q(k)\|_{\mathbb C^4}^2\mathd k\leq\|q\|_0^2.
\end{equation}
Since $e^{-\tau A(k)}\to 1$ pointwise as $\tau\to 0$, dominated convergence gives
$$\lim_{\tau\downarrow 0}\|T_\tau q-q\|_0^2=\lim_{\tau\downarrow 0}\int_{\mathbb R}\frac{|e^{-\tau A(k)}-1|^2}{2A(k)}\|\hat{q}(k)\|_{\mathbb C^4}^2\mathd k=0.$$
On account of Equation \eqref{Eq: Aux3}, $T_\tau$ extends uniquely by continuity to a contraction on $\mathcal{H}_0^{(1)}$. This property combined with strong continuity on $\mathcal{E}_0$, which is dense in $\mathcal{H}_0^{(1)}$, implies strong continuity of the extended semigroup. Furthermore, it is self-adjoint because its multiplier is real valued. If $q\in D(H_{1,0})$, dominated convergence gives
$$\lim_{\tau\downarrow 0}\left\|\frac{T_\tau q-q}{\tau}+A_xq\right\|_0=0.$$
Conversely, if the difference quotient converges in	$\mathcal{H}_0^{(1)}$, a pointwise convergent subsequence identifies its limit with $-A_xq$, and hence $A_xq\in\mathcal{H}_0^{(1)}$. 
To conclude, since the mass $m>0$, Equation \eqref{Eq: Aux4} entails that the graph norm of $H_{1,0}=A_x$ is equivalent to the $H^{1/2}$-norm. Since $\mathcal{S}(\mathbb{R};\mathbb{C}^2)$ is dense in $H^{1/2}(\mathbb{R};\mathbb{C}^2)$ and $P_+$ preserves both spaces, $\mathcal{E}_0$ is dense in $D(H_{1,0})$ with respect to the graph norm. Hence $\mathcal{E}_0$ is a core for $H_{1,0}$.
\end{proof}

We now restrict the preceding one-particle dynamics to the physical subspace identified in Section \ref{Sec: Physical one-particle space}. The commutation of $\mathbb{P}_{\mathrm{phys}}$ with $A_x$ ensures that this subspace is preserved under the Euclidean time-evolution semigroup.

\begin{corollary}\label{Cor: Physical 1-Particle Hamiltonian}
The operators $T_\tau$ preserve both $\mathcal E_{0,\phys}$ and	$\mathcal{H}_{\phys}^{(1)}$. The restrictions
$$T_\tau^{\phys}:\mathcal{E}_{0,\phys}\longrightarrow\mathcal{E}_{0,\phys}, \quad q\mapsto T_\tau^{\phys}q\assign e^{-\tau A_x}q,$$
extend uniquely to a strongly continuous self-adjoint contraction semigroup on $\mathcal{H}_{\phys}^{(1)}$. Its positive Hamiltonian is
$$H_{1,\phys}\assign H_{1,0}\big|_{\mathcal{H}_{\phys}^{(1)}},\qquad D(H_{1,\phys})=D(H_{1,0})\cap\mathcal{H}_{\phys}^{(1)}.$$
Equivalently, on the quotient $\mathcal{H}_0^{(1)}/\ker\mathbb{P}_{\phys}$, $T_\tau^{\phys}[q]=[T_\tau q]$.
\end{corollary}

\begin{proof}
Since $[\mathbb{P}_{\phys},A_x]=0$, by functional calculus it descends that $\mathbb{P}_{\phys}T_\tau=T_\tau\mathbb{P}_{\phys}$. Hence $\mathcal{E}_{0,\phys}$ and $\mathcal{H}_{\phys}^{(1)}$ are invariant under the action of $T_\tau$, while $\mathcal{H}_{\phys}^{(1)}$ is a reducing subspace. All the claims therefore follow by restriction from Proposition \ref{Prop: 1-Particle Hamiltonian}.
\end{proof}

\subsubsection{Markovian evolution and the Fock Hamiltonian}\label{Sec: Markovian reconstruction}
	
We now obtain the Fock evolution directly from the Euclidean time translation and conditional expectation. Set
$\mathcal{E}_0^{\wedge 0}=\mathcal{E}_{0,\phys}^{\wedge 0}\assign\mathbb{C}$ and $W_{0,0}^{\Psi_0}(z)\assign z\mathbb{I}_{\mathcal{A}}$. Define
\begin{align}
\mathcal{A}_0^{\alg}&\assign\bigoplus_{n=0}^{\infty}W_{n,0}^{\Psi_0}(\mathcal{E}_0^{\wedge n})\subset\mathcal{A}_{\leq 0},
\label{Eq: Algebraic time-zero Wick space}\\
\mathcal{A}_{\phys}^{\alg}&\assign\bigoplus_{n=0}^{\infty}W_{n,0}^{\Psi_0}(\mathcal{E}_{0,\phys}^{\wedge n})	\subset\mathcal{A}_{\leq 0},\label{Eq: Algebraic time-zero Wick space2}
\end{align}
where, in both expressions, we consider only the algebraic sum and the algebraic antisymmetric tensor product. They are the unital algebras generated by $\Psi_0(\mathcal{E}_0)$ and $\Psi_0(\mathcal{E}_{0,\phys})$, respectively. Indeed, the recursive Wick rule allows one to express Wick polynomials in terms of ordinary, non-Wick ordered monomials and, conversely, ordinary, non-Wick ordered monomials in terms of Wick polynomials.\\
	
\noindent We use the previous algebras to represent the Fermionic fields and their product. In order to obtain this representation we introduce the map 
$$\mathfrak{W}_0:\bigoplus_{n=0}^{+\infty}\mathcal{E}_0^{\wedge n}\longrightarrow\mathcal{A}_{0}^{\mathrm{alg}},\qquad (u_n)_n\longmapsto\sum_n W_{n,0}^{\Psi_0}(u_n),$$
where $(u_n)_n\in \bigoplus_{n=0}^{+\infty}\mathcal{E}_0^{\wedge n}$. 
\begin{lemma}
The map $\mathfrak W_0$ is bijective from $\bigoplus_{n=0}^{+\infty}\mathcal{E}_0^{\wedge n}$ into $\mathcal{A}_{0}^{\mathrm{alg}}$, and its restriction to $\bigoplus_{n=0}^{+\infty}\mathcal{E}_{0,\phys}^{\wedge n}$ is bijective onto $\mathcal{A}_{\phys}^{\mathrm{alg}}$.
\end{lemma}
\begin{proof}
Observe that $\mathfrak{W}_0$ is by construction surjective in both instances and, therefore, we need only to prove injectivity. Suppose that $\sum_{n=0}^N W_{n,0}^{\Psi_0}(u_n)=0$. Applying both sides to
$\Omega$ yields $\sum_{n=0}^N\mathcal{W}_{n,0}^{\Psi_0}(u_n)=0$. These vectors lie in mutually orthogonal particle sectors; their injectivity by Proposition \ref{Prop: Auxiliary Wick sectors} implies $u_n=0$ for every $n$. This proves the injectivity of $\mathfrak{W}_0$. 
\end{proof}
	
\noindent We introduce an evolution structure on $\mathcal{A}^{\alg}_0$ and $\mathcal{A}^{\alg}_{\phys}$ induced by the combination of the time translation $\Theta_{\tau}$ and the conditional expectation $\omega_0$. For $\tau\geq 0$, we define $K_\tau$ on elements of the form $W_{n,0}^{\Psi_0}(Q(F_1)\wedge\cdots\wedge Q(F_n))$ by
\begin{equation}\label{Eq: Markovian Wick evolution}
\begin{aligned}
&\mathcal{K}_\tau\left(W_{n,0}^{\Psi_0}(Q(F_1)\wedge\ldots\wedge Q(F_n))\right)\assign\omega_0\!\left(W_n^\psi(\Theta_\tau F_1\otimes\cdots\otimes\Theta_\tau F_n)\right),
\end{aligned}
\end{equation}
and we set $\mathcal{K}_\tau\mathbb{I}_{\mathcal{A}}=\mathbb{I}_{\mathcal{A}}$.
	
\begin{proposition}\label{Prop: Markovian Wick evolution}
The maps $\mathcal K_\tau$ as per Equation \eqref{Eq: Markovian Wick evolution} is well defined on $\mathcal{A}_0^{\alg}$, it leaves $\mathcal{A}_{\phys}^{\alg}$ invariant and it forms a linear semigroup which satisfies
\begin{equation}\label{Eq: Markovian intertwining}
\mathcal{K}_\tau W_{n,0}^{\Psi_0}(u)=W_{n,0}^{\Psi_0}(T_\tau^{\wedge n}u),\qquad u\in\mathcal E_0^{\wedge n},
\end{equation}
where $T_\tau^{\wedge0}=\mathbb I_{\mathbb C}$. Equivalently, on the algebraic direct sum,
\begin{equation}\label{Eq: Algebraic Fock intertwining}
\mathfrak{W}_0^{-1}\mathcal{K}_\tau\mathfrak{W}_0=\bigoplus_{n=0}^{\infty}T_\tau^{\wedge n}.
\end{equation}
\end{proposition}
	
\begin{proof}
For $q_j=Q(F_j)$, Proposition~\ref{Prop: Wick_vs_time_zero}, the definition of $\Psi_0$, and Equation \eqref{Eq: T_tau} give
\begin{align*}
\omega_0\!\left(W_n^\psi(\Theta_\tau F_1\otimes\cdots\otimes\Theta_\tau F_n)\right)&=W_n^{\psi_0}(\Theta_\tau F_1\otimes\cdots\otimes\Theta_\tau F_n)\\
&=W_{n,0}^{\Psi_0}(Q(\Theta_\tau F_1)\wedge \ldots \wedge Q(\Theta_\tau F_n))\\
&=W_{n,0}^{\Psi_0}(T_\tau q_1\wedge \ldots \wedge T_\tau q_n).
\end{align*}
The result depends only on $q_1\wedge\cdots\wedge q_n$, namely it is unchanged when any $F_j$ is replaced by another representative of $q_j$. Furthermore, by linearity it respects all relations
in $\mathcal{E}_0^{\wedge n}$. Equation \eqref{Eq: Algebraic time-zero Wick space} also ensures that the definition is independent of the chosen finite decomposition into Wick monomials. This proves that the map $\mathcal{K}_\tau$ is well-defined as per Equation \eqref{Eq: Markovian intertwining}. Equation \eqref{Eq: Algebraic Fock intertwining} implies $\mathcal{K}_0=\mathbb{I}$ and
$\mathcal{K}_{\tau+\sigma}=\mathcal{K}_\tau\mathcal{K}_\sigma$ for all $\tau,\sigma\geq 0$. Finally, $T_\tau$ preserves $\mathcal{E}_{0,\phys}$ by Corollary \ref{Cor: Physical 1-Particle Hamiltonian}, which implies that $\mathcal{A}_{\phys}^{\alg}$ is invariant.
\end{proof}
	
\noindent For $n\geq0$, set
\begin{equation}\label{Eq: Auxiliary n-particle semigroup}
T_{\tau,0}^{(0)}\assign\mathbb{I}_{\mathbb{C}},\qquad T_{\tau,0}^{(n)}\assign T_\tau^{\wedge n}\quad(n\geq1),
\end{equation}
and set $H_{0,0}\assign0$, while, for $n\geq1$, let $H_{n,0}$ be the closure of
\begin{equation}\label{Eq: Auxiliary n-particle Hamiltonian}
H_{n,0}(q_1\wedge\cdots\wedge q_n)\assign\sum_{r=1}^n q_1\wedge\cdots\wedge H_{1,0}q_r\wedge\cdots\wedge q_n,\qquad q_j\in\mathcal{E}_0.
\end{equation}

\begin{proposition}\label{Prop: Auxiliary Fock Hamiltonian}
With the vacuum norm $\|a\|_{\mathrm{W}}\assign\|a\Omega\|_{\mathcal F_-(V)}$,the completion of $\mathcal{A}_0^{\alg}$ is canonically	isomorphic to $\mathcal{F}_0$. Under this identification, $\mathcal{K}_\tau$ extends to the strongly continuous	self-adjoint contraction semigroup
\begin{equation}\label{Eq: Auxiliary Fock semigroup}
\mathbb{T}_{\tau,0}\assign\Gamma_-(T_\tau)=\bigoplus_{n=0}^{\infty}T_{\tau,0}^{(n)}.
\end{equation}
Its positive Hamiltonian is
\begin{equation}\label{Eq: Auxiliary Fock Hamiltonian}
\mathbb{T}_{\tau,0}=e^{-\tau H_0},\qquad H_0\assign\mathd\Gamma(H_{1,0})=\bigoplus_{n=0}^{\infty}H_{n,0},
\end{equation}
where $H_{n,0}=\mathd\Gamma_n(H_{1,0})$. The set $\bigoplus_{n=0}^{+\infty}\mathcal{E}_0^{\wedge n}$ is a core and we have, for $q_1,\ldots,q_n\in\mathcal E_0$,
\begin{align}
\mathbb{T}_{\tau,0}W_{n,0}^{\Psi_0}(q_1,\ldots,q_n)\Omega&=W_{n,0}^{\Psi_0}(T_\tau q_1,\ldots,T_\tau q_n)\Omega, \label{Eq: Fock semigroup on auxiliary Wick vectors}\\
H_0W_{n,0}^{\Psi_0}(q_1,\ldots,q_n)\Omega&=\sum_{r=1}^n W_{n,0}^{\Psi_0}(q_1,\ldots,H_{1,0}q_r,\ldots,q_n)\Omega.	\label{Eq: Fock Hamiltonian on auxiliary Wick vectors}
\end{align}
\end{proposition}
	
\begin{proof}
For $a=\mathfrak W_0((u_n)_n)$,	Proposition~\ref{Prop: Auxiliary Wick sectors} and orthogonality of	particle sectors give
$$\|a\|_{\mathrm{W}}^2=\sum_n\|u_n\|^2,\qquad\|\mathcal{K}_\tau a\|_{\mathrm{W}}^2=\sum_n\|T_\tau^{\wedge n}u_n\|^2		\leq\|a\|_{\mathrm W}^2.$$
Thus $\mathfrak{W}_0^{-1}$ extends to the stated isometric identification, and Equation \eqref{Eq: Algebraic Fock intertwining} becomes Equation \eqref{Eq: Auxiliary Fock semigroup}. More precisely, on Wick vectors in $\mathcal{F}_-(V)$,
$$(\mathcal{K}_\tau a)\Omega=\Gamma_-(J_0)\mathbb{T}_{\tau,0}\Gamma_-(J_0)^{-1}(a\Omega),$$
with the inverse taken on the range of $\Gamma_-(J_0)$.	Each $T_\tau^{\wedge n}$ is a self-adjoint contraction and it is strongly continuous on $\mathcal H_0^{(n)}$. To conclude, the contraction bound extends strong continuity from finite-particle vectors to the whole $\mathcal{F}_0$.\\
		
\noindent Differentiating $T_\tau^{\wedge n}$ at $\tau=0$ on $\mathcal{E}_0^{\wedge n}$ gives minus the sum in Equation \eqref{Eq: Auxiliary n-particle Hamiltonian}. Since	$\mathcal{E}_0$ is a core for $H_{1,0}$ as per Proposition \ref{Prop: 1-Particle Hamiltonian}, these finite wedge products form a core for $H_{n,0}=\mathd\Gamma_n(H_{1,0})$.	The generator of the direct-sum semigroup is consequently $-H_0$ with core $\oplus_{n=0}^{+\infty}\mathcal{E}_0^{\wedge n}$.  The identities on Wick vectors follow from Equation \eqref{Eq: Wick vector exterior product}.
\end{proof}
	
\noindent Define the physical Fock projection by
\begin{equation}\label{Eq: Fock physical projection}
\Pi_{\phys}\assign\Gamma_-(\mathbb{P}_{\phys})=\bigoplus_{n=0}^{\infty}\mathbb{P}_{\phys}^{\wedge n},\qquad		\operatorname{Ran}\Pi_{\phys}=\mathcal{F}_{\phys}.
\end{equation}
	
\begin{corollary} \label{Prop: Physical Fock Hamiltonian}
The restriction of $\mathcal{K}_\tau$ to $\mathcal{A}_{\phys}^{\alg}$ extends in the vacuum norm to
\begin{equation}\label{Eq: Physical Fock semigroup}
\mathbb{T}_{\tau,\phys}\assign\mathbb{T}_{\tau,0}|_{\mathcal{F}_{\phys}}=\Gamma_-(T_\tau^{\phys})		=\bigoplus_{n=0}^{\infty}(T_\tau^{\phys})^{\wedge n}.
\end{equation}
It is a strongly continuous self-adjoint contraction semigroup with
\begin{equation}\label{Eq: Physical Fock Hamiltonian}
\mathbb{T}_{\tau,\phys}=e^{-\tau H_{\phys}},\qquad H_{\phys}\assign\mathd\Gamma(H_{1,\phys})=H_0|_{\mathcal{F}_{\phys}}.
\end{equation}
The auxiliary Wick-vector identities above restrict to	$q_j\in\mathcal{E}_{0,\phys}$ with $T_\tau,H_{1,0},H_0$
replaced by $T_\tau^{\phys},H_{1,\phys},H_{\phys}$,	respectively.
\end{corollary}
	
\begin{proof}
The relation $\mathbb{P}_{\phys}T_\tau=T_\tau\mathbb{P}_{\phys}$ implies $\Pi_{\phys}\mathbb{T}_{\tau,0}=\mathbb{T}_{\tau,0}\Pi_{\phys}$. Thus $\mathcal{F}_{\phys}$ is reducing for the auxiliary semigroup and its generator, proving all the statements by restriction.
\end{proof}

\appendix
\section{Appendix}\label{Appendix}

We present the proof of Proposition \ref{Prop: Wick_vs_time_zero}. We also repeat the statement for the sake of readability.

\begin{proposition*}
For any $n\in\mathbb{N}$ and for any $F_1,\dots F_n\in\mathcal S(\mathbb R^2_+;\mathbb C^4)$, 
$$\omega_0\bigl(W^\psi_n(F_1\otimes\cdots\otimes F_n)\bigr)=W^{\psi_0}_n(F_1\otimes\cdots\otimes F_n).$$
\end{proposition*}

\begin{proof}
We introduce the notations
\begin{equation}\label{Eq: Useful Notation}
\psi_>(F) \assign \psi(F)-\omega_0 (\psi(F))=\psi(F) -\psi_0(F),\quad\forall F\in\mathcal{S}(\mathbb{R}^2_+;\mathbb{C}^4)
\end{equation}
and 
\begin{equation}\label{Eq: Useful Notation 2}
W^\psi_n (F_{1>}\otimes\cdots\otimes F_n)\assign\Wick{\psi_>(F_1)\cdots\psi(F_n)}\quad\forall F_1,\dots,F_n\in\mathcal{S}(\mathbb{R}^2_+;\mathbb{C}^4), 
\end{equation}
where $F_{1>}$ is a convenient notation and it should not be read as an element of $\mathcal{S}(\mathbb{R}^2_+; \mathbb{C}^4)$. Observe that, per construction, $\omega_0(\psi_>(F))=0$. Bearing this property in mind, we prove the statement by induction on the number of fields in the Wick monomials. Starting from $n=1$, Equations \eqref{Eq: Wick Polynomials order 1 and 2} and \eqref{Eq: Wick Polynomials order 1 and 2 - zero time}, combined with the notations introduced in Equations \eqref{Eq: Wick Polynomials vs :} and \eqref{Eq: Wick Polynomials vs : - zero time}, yield
$$\omega_0 (W_1^\psi(F)) = \omega_0 (\psi (F)) = \psi_0 (F) = W_1^{\psi_0}(F).$$

\noindent Following the inductive hypothesis we can assume the statement to be true at order $n-1$ in the number of fields. Using Equation \eqref{Eq: Wick Polynomials order 3} which is valid also for $n=2$, in combination with the linearity of $\omega_0$, it holds that
  \begin{equation}
    \begin{array}{lll}
      \omega_0(W^\psi_n(F_1\otimes\cdots\otimes F_n)) & = & \omega_0 [\psi
      (F_1) W^\psi_{n-1} (F_2\otimes\cdots\otimes F_n)] +\\
      &  & - \sum_{j = 2}^n (-1)^j \omega (\psi(F_1)\psi(F_j)) \omega_0
      (W^\psi_{n - 2} (F_2\otimes\cdots\widetilde{F_j}\otimes \cdots\otimes F_n))\\
      & = & \omega_0 [\psi (F_1) W^\psi_{n-1} (F_2\otimes\cdots\otimes F_n)]
      +\\
      &  & - \sum_{j = 2}^n (- 1)^j \omega (\psi (F_1) \psi (F_j)) W^\psi_{n - 2}
      (F_2\otimes\cdots\otimes\widetilde{F_j}\otimes\cdots\otimes F_n)\\
      & = & \omega_0 [\psi (F_1) W^\psi_{n-1} (F_2\otimes\cdots\otimes F_n)]
      +\\
      &  & - \sum_{j = 2}^n (- 1)^j \omega (\psi_0 (F_1) \psi_0 (F_j)) W^{\psi_0}_{n-2} (F_2\otimes\cdots\otimes\widetilde{F_j}\otimes\cdots F_n) +\\
      &  & - \sum_{j = 2}^n (- 1)^j \omega (\psi_> (F_1) \psi_> (F_j)) W^{\psi_0}_{n-2} (F_2\otimes\cdots\otimes\widetilde{F_j}\otimes\cdots F_n)
    \end{array} \label{Eq:1}
  \end{equation}
  where the tilde denotes a removed element. Herein we exploited the inductive hypothesis, that $\omega_0$ defines a conditional expectation as per Definition \ref{Def: Conditional Expectations}, as well as that, as a by-product, of Definition \ref{Def: Grassmannian White Noise}
  \begin{equation}
    \omega (\psi_0 (F_1) \psi_>(F_j)) = 0. \label{Eq:indep}
  \end{equation}
  Focusing on the first term, the properties of conditional expectations and Equation \eqref{Eq: Useful Notation} entail
  \begin{equation}
    \begin{array}{lll}
      \omega_0 [\psi (F_1) W^\psi_{n-1} (F_2\otimes\cdots\otimes F_n)] & = & \omega_0 [\psi_0 (F_1) W^\psi_{n-1} (F_2\otimes\cdots\otimes F_n)]
      + \omega_0 [\psi_> (F_1) W^\psi_{n-1} (F_2\otimes\cdots\otimes F_n)]\\
      & = & \psi_0 (F_1) W^{\psi_0}_{n-1} (F_2\otimes\cdots\otimes F_n) +
      \omega_0 [\psi_>(F_1)W^\psi_{n-1} (F_2\otimes\cdots\otimes F_n)],
    \end{array} \label{Eq:2}
  \end{equation}
where, employing the identity $\omega (\psi_>(F)\psi(F^\prime))= \omega (\psi_> (F)\psi_>(F^\prime))$ for all $F,F^\prime\in\mathcal{S}(\mathbb R^2_+;\mathbb C^4)$,
 \begin{equation}\label{Eq:3}
  \begin{array}{lll}
\omega_0 [\psi_>(F_1) W^\psi_{n-1} (F_2\otimes\cdots\otimes F_n)] & = &  \omega_0[W^\psi_n(F_{1>}\otimes F_2\otimes\cdots\otimes F_n)]\\
      &  & + \sum_{j=2}^n (-1)^j \omega(\psi_> (F_1)\psi (F_j)) \omega_0
      [W^\psi_{n-2} (F_2\otimes\cdots\otimes\widetilde{F_j}\otimes\cdots\otimes F_n)]\\
      & = & \omega_0 [W^\psi_n (F_{1>}\otimes F_2\otimes\cdots\otimes F_n)]\\
      &  & + \sum_{j = 2}^n (-1)^j \omega (\psi_>(F_1)\psi_>(F_j)) W^{\psi_0}_{n-2} (F_2\otimes\cdots\otimes\widetilde{F_j}\otimes\cdots\otimes F_n).
    \end{array}
  \end{equation}
Inserting Equations~\eqref{Eq:2} and~\eqref{Eq:3} in  Equation~\eqref{Eq:1}, we end up with
  \[ \begin{array}{lll}
  \omega_0\bigl(W^\psi_n(F_1\otimes\cdots\otimes F_n)\bigr) & = & \psi_0 (F_1) W^{\psi_0}_{n-1} (F_2\otimes\cdots\otimes F_n) + \omega_0[W^\psi_n (F_{1>}\otimes F_2\otimes\cdots\otimes F_n)] +\\
       &  & + \sum_{j=2}^n (-1)^j \omega(\psi_>(F_1)\psi_>(F_j)) W^{\psi_0}_{n-2} (F_2\otimes\cdots\otimes\widetilde{F_j}\otimes\cdots\otimes F_n)\\
       &  & - \sum_{j=2}^n (-1)^j \omega(\psi_0 (F_1) \psi_0 (F_j)) W^{\psi_0}_{n-2} (F_2\otimes\cdots\otimes\widetilde{F_j}\otimes\cdots\otimes F_n) +\\
       &  & - \sum_{j=2}^n (- 1)^j \omega(\psi_>(F_1)\psi_>(F_j)) W^{\psi_0}_{n-2} (F_2\otimes\cdots\otimes\widetilde{F_j}\otimes\cdots\otimes F_n)\\
       & = & \psi_0 (F_1) W^{\psi_0}_{n-1} (F_2\otimes\cdots\otimes F_n)  + \omega_0 [W^\psi_n (F_{1>}\otimes F_2\otimes\cdots\otimes F_n)]\\
       &  & - \sum_{j=2}^n (-1)^j \omega (\psi_0(F_1)\psi_0(F_j)) W^{\psi_0}_{n-2} (F_2\otimes\cdots\otimes\widetilde{F_j}\otimes\cdots\otimes F_n) +\\
       & = & W^{\psi_0}_n (F_1\otimes\cdots\otimes F_n) + \omega_0 [W^{\psi}_n(F_{1>}\otimes F_2\otimes\cdots\otimes F_n)].
     \end{array} \]
In order to tame the last term, we observe that
  \[ - \omega_0 [W^{\psi}_n(F_{1>}\otimes F_2\otimes\cdots\otimes F_n)] = \omega_0 [W^{\psi}_n(F_{2>}\otimes F_1\otimes\cdots\otimes F_n)].\]
If we iterate the above procedure, we obtain on account of the Wick rule,
  \begin{gather}
      \omega_0 [W^\psi_n (F_2\otimes F_{1>}\otimes\cdots\otimes F_n)]=\notag\\
      \omega_0 [\psi(F_2) W^\psi_{n-1} (F_{1>}\otimes\widetilde{F_2}\otimes\cdots\otimes F_n)] -
      \omega (\psi (F_2)\psi_>(F_1)) \omega_0 (W^\psi_{n-2} (F_3\otimes\cdots\otimes F_n)\notag\\
- \sum_{j = 3}^n (- 1)^{j-1} \omega(\psi(F_2)\psi(F_j))\omega_0(W^\psi_{n-2}(F_{1>}\otimes \widetilde{F_2}\otimes F_3\otimes\cdots\otimes\widetilde{F_j}\otimes\cdots\otimes F_n)) .
\label{Eq:4}
  \end{gather}
Using the inductive hypothesis and $\omega_0(\psi_>(F))=0$ for all $F\in\mathcal{S}(\mathbb R^2_+;\mathbb C^4)$,
\begin{gather*}
 \omega_0(W^\psi_{n-2}(F_3\otimes\cdots\otimes F_n)) = W^{\psi_0}_{n-2}(F_3\otimes\cdots\otimes F_n),\\
 \omega_0(W^\psi_{n-2}(F_{1>}\otimes \widetilde{F_2}\otimes F_3\otimes\cdots\otimes\widetilde{F_j}\otimes\cdots\otimes F_n)) = W^{\psi_0}_{n-2}(\omega_0 (\psi_> (F_1))F_3\otimes\cdots\widetilde{F_j}\otimes\cdots\otimes F_n)=0.
\end{gather*}
Using once more $\omega_0(\psi_>(F))=0$, $\omega (\psi_>(F)\psi(F^\prime))= \omega (\psi_> (F)\psi_>(F^\prime))$ for all $F,F^\prime\in\mathcal{S}(\mathbb R^2_+;\mathbb C^4)$, Equation~\eqref{Eq:4} reads
\begin{gather}
\omega_0 [W^\psi_n (F_2\otimes F_{1>}\otimes\cdots\otimes F_n)] =\notag \\
\omega_0 [\psi (F_2) W^\psi_{n-1}(F_{1>}\otimes\widetilde{F_2}\otimes\cdots\otimes F_n)] - \omega(\psi (F_2)\psi_>(F_1))W^{\psi_0}_{n-2}(F_3\otimes\cdots\otimes F_n)\notag\\
=  \omega_0 [\psi_>(F_2) W^\psi_{n-1} (F_{1>}\otimes \widetilde{F_2}\otimes F_3\otimes\cdots\otimes F_n)]-\omega(\psi_>(F_2)\psi_>(F_1)) W^{\psi_0}_{n-2}(F_3\otimes\cdots\otimes F_n),\label{Eq:5}
\end{gather}
 where we used implicitly the following chain of identities:
\begin{gather*}
\omega_0 [\psi (F_2) W^\psi_{n-1}(F_{1>}\otimes\widetilde{F_2}\otimes\cdots\otimes F_n)] =\\
       \omega_0 [\psi_>(F_2)W^\psi_{n-1}(F_{1>}\otimes\widetilde{F_2}\otimes\cdots\otimes F_n)] +
       \omega_0 [\psi_0 (F_2) W^\psi_{n-1} (F_{1>}\otimes\widetilde{F_2}\otimes\cdots\otimes F_n)]\\
        =\omega_0 [\psi_>(F_2)W^\psi_{n-1}(F_{1>}\otimes\widetilde{F_2}\otimes\cdots\otimes F_n)] +
\psi_0 (F_2)  \omega_0 [W^\psi_{n-1} (F_{1>}\otimes\widetilde{F_2}\otimes\cdots\otimes F_n)]\\
       =\omega_0 [\psi_>(F_2)W^\psi_{n-1}(F_{1>}\otimes\widetilde{F_2}\otimes\cdots\otimes F_n)] + \psi_0 (F_2) W^{\psi_0}_{n-1}(F_{1>}\otimes\widetilde{F_2}\otimes\cdots\otimes F_n)
       )\\
       = \omega_0 [\psi_>(F_2)W^\psi_{n-1}(F_{1>}\otimes\widetilde{F_2}\otimes\cdots\otimes F_n)],
\end{gather*}
where $W^{\psi_0}_{n-1}(F_{1>}\otimes\widetilde{F_2}\otimes\cdots\otimes F_n)$ vanishes as a consequence of $\psi_0(F_{1>})=0$. Focusing on the first term on the right hand side of Equation~\eqref{Eq:5}, we can tame it using the chain of identities
\begin{gather*}
\omega_0 [W^\psi_n (F_2\otimes F_{1>}\otimes\cdots\otimes F_n)] =\\
 \omega_0 [W^\psi_n (F_{2>}\otimes F_{1>}\otimes\cdots\otimes F_n)] + \omega
       (\psi_>(F_2) \psi_> (F_1)) \omega_0 (W^\psi_{n-2}(F_3\otimes\cdots\otimes F_n))\\
        + \sum_{j = 3}^n (-1)^{j-1} \omega(\psi_> (F_2) \psi (F_j))
       \omega_0 (W^\psi_{n-2} (F_{1>}\otimes\widetilde{F_2}\otimes\cdots\otimes\widetilde{F_j}\otimes\cdots\otimes F_n))\\
        - \omega (\psi_> (F_2) \psi_> (F_1)) W^{\psi_0}_{n-2} (F_3\otimes\cdots\otimes F_n)=\omega_0 [W^\psi_n (F_{2>}\otimes F_{1>}\otimes F_3\otimes\cdots\otimes F_n)],
\end{gather*}
 where in the last step we used that, on account of the inductive hypothesis,
$$ \omega (\psi_> (F_2) \psi_> (F_1)) \omega_0 (W^\psi_{n-2}(F_3\otimes\cdots\otimes F_n))=\omega(\psi_>(F_2)\psi_>(F_1)) W^{\psi_0}_{n-2}(F_3\otimes\cdots\otimes F_n),$$
while all the terms associated to the sum over $j$ are vanishing using the inductive hypothesis in combination with $\omega_0(\psi_>(F_1)) = 0$. To summarize, 
$$\omega_0 [W^\psi_n (F_{1>}\otimes F_2\otimes\cdots\otimes F_n)] = -\omega_0[W^\psi_n(F_{2>}\otimes F_{1>}\otimes\cdots\otimes F_n)].$$
Iterating this procedure we end up with
$$\omega_0 [W^\psi_n (F_{1>}\otimes F_2\otimes\cdots\otimes F_n)]  = (-1)^{n-1} \omega_0 [W^\psi_n (F_{n>}\otimes\cdots\otimes F_{1>})]=0,$$
where in the last equality we used the defining properties of the zero-time Fermionic Wick polynomials. In turn this yields
$$\omega_0 [W^\psi_n(F_{1>}\otimes F_2\otimes\cdots\otimes F_n)] = 0,$$
which concludes the proof.
\end{proof}

\begin{acknowledgments*}
	 CD and PR acknowledge the support of the INFN Sezione di Pavia and of Gruppo Nazionale di Fisica Matematica, part of INdAM. PF acknowledges the support of the INFN Sezione di Genova and of Gruppo Nazionale di Fisica Matematica, part of INdAM. Part of this work has appeared in the Master thesis of PF submitted on the 19/09/25 as a partial fulfillment of the requirement to obtain a Master's degree in Physics at the University of Pavia.
\end{acknowledgments*}

\end{document}